\newcommand{\esp}[2][\mathbb E] {#1\left[#2\right]}
\newcommand{\expon}[1]{\exp\left(#1\right)}

\documentclass[11pt,a4paper,english]{article}

\usepackage{amsfonts}
\usepackage{amssymb,amsmath}
\usepackage{color}
\usepackage{graphicx}
\usepackage{amsmath}
\usepackage{enumerate}
\usepackage{hyperref}
\usepackage{nameref}

\usepackage{natbib}

\makeatletter
\let\orgdescriptionlabel\descriptionlabel
\renewcommand*{\descriptionlabel}[1]{
  \let\orglabel\label
  \let\label\@gobble
  \phantomsection
  \edef\@currentlabel{#1}
  \let\label\orglabel
  \orgdescriptionlabel{#1}
}

\numberwithin{equation}{section}

\makeatother

\newtheorem{theorem}{Theorem}[section]

\newtheorem{corollary}[theorem]{Corollary}

\newtheorem{definition}[theorem]{Definition}
\newtheorem{example}[theorem]{Example}

\newtheorem{lemma}[theorem]{Lemma}

\newtheorem{proposition}[theorem]{Proposition}

\newtheorem{remark}[theorem]{Remark}

\newenvironment{proof}[1][Proof]{\textbf{#1.} }{\ \rule{0.5em}{0.5em}}

\title{A Unified Approach to Systemic Risk Measures via Acceptance Sets}
\date{\today }

\author{Francesca Biagini\thanks{Department of Mathematics, University of Munich, Theresienstra{\ss}e 39, 80333 Munich, Germany. {\em francesca.biagini@math.lmu.de}.}
\and Jean-Pierre Fouque \thanks{Department of Statistics \& Applied Probability,
 University of California,
        Santa Barbara, CA 93106-3110, {\em fouque@pstat.ucsb.edu}. Work  supported by NSF grants
DMS-1107468 and DMS-1409434.} \and Marco Frittelli\thanks{Dipartimento di Matematica,
Universit\`a degli Studi di Milano,
Via Saldini 50,
20133 Milano, Italy, {\em marco.frittelli@unimi.it}. } \and
Thilo Meyer-Brandis\thanks{Department of Mathematics, University of Munich, Theresienstra{\ss}e 39, 80333 Munich, Germany.{\em meyerbr@math.lmu.de}. Francesca Biagini and Thilo Meyer-Brandis thank the Frankfurt Institute for Risk Management and
Regulation and  the Europlace Institute of Finance for support. \newline
Part of this research was performed while F. Biagini, M. Frittelli and T.
Meyer-Brandis were visiting the University of California Santa Barbara, and
the paper was finalized while the four authors were visiting the Institute
for Pure and Applied Mathematics (IPAM), which is supported by the National
Science Foundation.}}

\begin{document}

\maketitle

\begin{abstract}
The financial crisis has dramatically demonstrated that the traditional approach to apply univariate monetary risk measures to single institutions does not capture sufficiently the perilous systemic risk that is generated by the interconnectedness of the system entities and the corresponding contagion effects. This has brought awareness of the urgent need for novel approaches that capture systemic riskiness. The purpose of this paper is to specify a general methodological framework that is flexible enough to cover a wide range of possibilities to design systemic risk measures via multi-dimensional acceptance sets and aggregation functions, and to study corresponding examples. Existing systemic risk measures can usually be interpreted as the minimal amount of cash needed to secure the system \emph{after aggregating} individual risks. In contrast, our approach also includes systemic risk measures that can be interpreted as the minimal amount of cash  that secures the aggregated system by allocating capital to the single institutions \emph{before aggregating} the individual risks. This allows for a possible ranking of the institutions in terms of systemic riskiness measured by the optimal allocations. 
Another important feature of our
approach is the possibility of allocating cash according to the future
state of the system (scenario-dependent allocation). We illustrate with
several examples the advantages of
this feature. We also provide conditions which ensure monotonicity,
convexity, or quasi-convexity properties of our systemic risk measures.

\end{abstract}

\vspace{2mm}

\section{Introduction}

A large part of the current literature on systemic financial risk is concerned
with the modeling structure of financial networks and the analysis of the
contagion and the spread of a potential exogenous (or even endogenous)
shock into the system. For a given financial (possibly random) network and
a given random shock one then determines the \textquotedblleft
cascade\textquotedblright\ mechanism which generates possibly many defaults.
This mechanism often requires a detailed description of the balance sheet of
each institution; assumptions on the interbank network and exposures, on the
recovery rate at default, on the liquidation policy; the analysis of direct
liabilities, bankruptcy costs, cross-holdings, leverage structures, fire
sales, and liquidity freezes. 

Among the many contributions we mention here the classical contagion model
proposed by \cite{EisenbergNoe}, the default model of \cite{GaiKapadia}, the illiquidity cascade models of \cite{GaiKapadia2}, \cite{HurdCellaiMelnikShao} and \cite{SeungHwanLee}, the asset fire
sale cascade model by \cite{CifuentesFerrucciShin} and \cite{CaccioliShresthaMooreFarmer}, as well as the  model  in \cite{AwiszusWeber} that additionally includes cross-holdings. For an exhaustive reference on the literature we defer the
reader to the recent volume: \textquotedblleft Contagion! The Spread of
Systemic Risk in Financial Networks\textquotedblright, \cite{Hurd}. 

These approaches may be relevant also from the viewpoint of a policy maker
that has to intervene and regulate the banking system to reduce the risk
that, in case of an adverse (local) shock, a substantial part or even the
complete system breaks down. 

However, once such a model for the financial network has been identified and
the mechanism for the spread of the contagion determined, one still has to
understand how to compare the possible final outcomes in a reasonable way
or, in other words, how to measure the risk carried by the global financial
system.  This is the focus of our approach, as we measure the risk embedded
in a financial system taking as primitive a vector $\mathbf{X}=(X^{1},\ldots
,X^{N})$ of positions, where $X^{i}$ represents the position of
institution $i$.
Our approach is very close in spirit to the \textquotedblleft
classical\textquotedblright\ conceptual framework initiated by the seminal
paper by \cite{ArtznerDelbaenEberHeath} and that has been recently adopted also to analyze systemic
risk by \cite{ChenIyengarMoallemi}, \cite{KromerOverbeckZilch} and \cite{HoffmannMeyer-BrandisSvindland}. \newline
We recall this classical approach, in the case of one single institution, by
the following two quotes from \cite{ArtznerDelbaenEberHeath}:

\

\textquotedblleft \textit{The basic objects of our study shall therefore be
the random variables on the set of states of nature at a future date,
interpreted as possible future values of positions or portfolios currently
held.\textquotedblright }

...

\textit{\textquotedblleft These measures of risk can be used as (extra)
capital requirements to regulate the risk assumed by market participants,
traders, and insurance underwriters, as well as to allocate existing
capital.\textquotedblright }

\bigskip

\noindent Of course one main difference is that we have to take into consideration not
just one single institution but the global system and in this paper we will
illustrate how to achieve this in an appropriate way. 
We interpret $X^{i}$ as the profits and losses of institution $i$ at a future time $T$, 
precisely as the gain if $X^{i}$ is
positive or as the loss if $X^{i}$ is negative. Such  profit and loss is typically uncertain and therefore it will be modeled by
a random variable $X^{i}(\omega )$ on some space of possible scenarios $\omega \in \Omega$. 

To summarize, we consider  the random vector $\mathbf{X}=$ $(X^{1},\ldots,X^{N})$ as primitive: One may interpret it as a  \textquotedblleft reduced form model\textquotedblright\  of a complex financial system without reference to a specific structural network model, and consequently $\mathbf{X}$ may already comprehend the potential risk of a contagion spread into the system. Otherwise one may also interpret $\mathbf{X}$ as the net worth of the positions before the contagion takes place and then the contagion mechanism will be embedded in the risk measure via an aggregation function (as in the model of \cite{EisenbergNoe}).  Either way, our scope is to provide a consistent criterion to asses wether one possible vector $\mathbf{X}$ is riskier than another.

\subsection{From one-dimensional to N-dimensional risk profiles}

In this subsection we review the literature on risk measurement based on acceptable sets, both in the traditional one-dimensional
setting as well as in the case of $N$ interacting financial institutions.
Here we denote with $\mathcal{L}^{0}(\mathbb{R}^{N}):=\mathcal{L}^{0}(\Omega
,\mathcal{F};\mathbb{R}^{N})$, $N\in \mathbb{N}$, the space of $\mathbb{R}%
^{N}$-valued random variables on the probability space $(\Omega ,\mathcal{F},\mathbb{P})$.

Traditional risk management strategies of financial systems have predominantly focused on the solvency of individual
institutions as if they were in isolation. A typical approach is to evaluate the risk $\eta (X^i)$ of each institution $i\in \{1,...N\}$ by applying a  \textit{univariate monetary risk measure} $\eta$ to the single financial positions. A monetary risk measure (see \cite{FollmerSchied2}) is a map $\eta :$ $\mathcal{L}^{0}(\mathbb{R})\rightarrow \mathbb{R}$  that can be interpreted as the minimal capital needed to secure a financial position with payoff $X\in \mathcal{L}^{0}(\mathbb{R})$,  i.e.~the minimal amount $m\in \mathbb{R}$ that must be added to $X$ in order to make the resulting (discounted) payoff at time $T$ acceptable: 
\begin{equation}
\eta (X):= \inf \{m\in \mathbb{R}\mid X+m\in \mathbb{A}\},  \label{1}
\end{equation}%
where the acceptance set $\mathbb{A}%
\subseteq \mathcal{L}^{0}(\mathbb{R})$ is assumed to be monotone, i.e. $%
X\geq Y\in \mathbb{A}$ implies $X\in \mathbb{A}$. In addition to decreasing
monotonicity, the characterizing feature of these maps is the cash
additivity property:
\begin{equation}\label{CAdd}
\eta (X+m)=\eta (X)-m\text{, for all }m\in \mathbb{R}\text{.}
\end{equation}%
Under the assumption that the set $\mathbb{A}$ is convex (resp. is a convex
cone) the maps in (\ref{1}) are convex (resp. convex and positively
homogeneous) and are called \textit{convex} (resp. coherent) \textit{risk
measures}, see \cite{ArtznerDelbaenEberHeath}, \cite{FollmerSchied}, \cite{FrittelliRosazza}. The principle that diversification should not
increase the risk is mathematically translated not necessarily with the
convexity property but with the weaker condition of quasiconvexity:
\begin{equation*}
\eta (\lambda X+(1-\lambda )Y)\leq \eta (X)\vee \eta (Y).
\end{equation*}%
As a result, in \cite{Cerreia-Vioglio} and \cite{FrittelliMaggis}, the only properties assumed in the definition of a
\textit{quasi-convex risk measure} are monotonicity and quasiconvexity. Such
risk measures can always be written as:%
\begin{equation}
\eta (X):= \inf \{m\in \mathbb{R}\mid X\in \mathbb{A}^{m}\},
\label{11}
\end{equation}%
where each set $\mathbb{A}^{m}\subseteq \mathcal{L}^{0}(\mathbb{R})$ is
monotone and convex, for each $m$. Here $\mathbb{A}^{m}$ is interpreted as
the class of payoffs carrying the same risk level $m$. Contrary to the
convex, cash additive case where each random variable is binary cataloged as
acceptable or as not acceptabe, in the quasi-convex case one admits various
degrees of acceptability, described by the risk level $m$, see \cite{ChernyMadan}.
Furthermore, in the quasi-convex case the cash additivity property will not
hold in general and one looses a direct interpretation of $m$ as the minimal
capital required to secure the payoff $X$, but preserves the interpretation
of $\mathbb{A}^{m}$ as the set of positions acceptable for the given risk
level $m$. By selecting $\mathbb{A}^{m}:=\mathbb{A-}m$, the risk measure in (%
\ref{1}) is clearly a particular case of the one in (\ref{11}).

However, the financial crisis has dramatically demonstrated that the traditional approach to apply univariate monetary risk measures to single institutions does not capture sufficiently the perilous systemic risk that is generated by the interconnectedness of the system entities and the corresponding contagion effects. This has brought awareness of the urgent need for novel approaches that capture systemic riskiness, and a rapidly growing literature
is concerned with designing more appropriate risk measures for financial systems. A \textit{systemic risk measure} is then a map $\rho :$ $\mathcal{L}^{0}(\mathbb{R}^N)\rightarrow \mathbb{R}$ that evaluates the risk $\rho (\mathbf{X})$ of the complete system $\mathbf{X}$ of financial positions. Most of the systemic risk measures in the existing literature are of the form
\begin{equation}
\rho (\mathbf{X})=\eta (\Lambda (\mathbf{X})),  \label{112}
\end{equation}%
where $\eta :$ $\mathcal{L}^{0}(\mathbb{R})\rightarrow \mathbb{R}$  is a
univariate risk measure and
\begin{equation*}
\Lambda :\mathbb{R}^{N}\rightarrow \mathbb{R}
\end{equation*}
is an aggregation rule that aggregates the $N$-dimensional risk factor $%
\mathbf{X}$ into a univariate risk factor $\Lambda (\mathbf{X})$
representing the total risk in the system. Some examples of aggregation rules found in the literature are the following:

\begin{itemize}
\item In general, one of the most common ways to aggregate multivariate risk is to simply sum the single risk factors: $\Lambda (\mathbf{x})=\sum_{i=1}^{N}x_{i}, \,\mathbf{x}=(x_1,...,x_N) \in \mathbb{R}^{N}$. Also in the literature on systemic risk measures there are examples using this aggregation rule, like for example the \textit{Systemic Expected Shortfall} introduced in \cite{Acharya}, or the \textit{Contagion Value at Risk (CoVaR)} introduced in \cite{AdrianBrunnermeier}. However, while summing up profit and loss positions might be reasonable from the viewpoint of a portfolio manager where the portfolio components compensate each other, this aggregation rule seems inappropriate for a financial system where cross-subsidization between institutions is rather unrealistic. Further, if the sum was a suitable aggregation of risk in financial systems, then the traditional approach of applying a univariate coherent risk measure $\eta$ to the single risk factors would be sufficiently prudential in the sense that by sub-linearity it holds that $\eta(\sum_{i=1}^{N}X_{i}) \le \sum_{i=1}^{N}\eta(X_{i})$.

\item One possible aggregation that takes the lack of cross-subsidization between financial institutions into account is to sum up losses only: $\Lambda (\mathbf{x})=\sum_{i=1}^{N}-x_{i}^{-}$. This kind of aggregation is for example used in \cite{Huang}, \cite{Lehar}. See also \cite{BrunnermeierCheridito} for an extension of this type of aggregation rule that also considers a certain effect of gains, as $\Lambda (\mathbf{x})=\sum_{i=1}^{N}-\alpha_i x_{i}^{-} + \sum_{i=1}^{N}\beta_i(x_{i}-v_i)^{+}$ for some $\alpha_i, \beta_i, v_i \in\mathbb{R}^{+}$, $i=1,\cdots,N$.


\item Beside the lack of cross-subsidization in a financial system, the
aggregation rule may also accounts for contagion effects that can
considerably accelerate systemwide losses resulting from an initial shock.
Motivated by the structural contagion model of \cite{EisenbergNoe}, in \cite%
{ChenIyengarMoallemi} they introduce an aggregation function that explicitly
models the net systemic cost of the contagion in a financial system by
defining the aggregation rule 

\begin{equation*}
\Lambda _{CM}(x)=\min_{y_{i}\geq x_{i}+\sum_{j=1}^{N}\Pi
_{ij}y_{j},\forall i=1,\cdots ,N,\,y\in \mathbb{R}_{+}^{N}}\left\{
\sum_{i=1}^{N}y_{i}\right\} \,.
\end{equation*}%
Here,  $\Pi =(\Pi _{ij})_{i,j=1,\cdots ,N}$ represents the relative liability matrix, i.e.firm $i$ has to pay the proportion $\Pi _{ij}$ of its total
liabilities to firm $j$. \\
In the literature there are various extensions of the structural contagion model of \cite{EisenbergNoe} and the corresponding aggregation rule that take into account further contagion channels of systemic risk such as effects from firesales or liquidity freezes, see e.g. \cite{AminiFilipovicMinca}, \cite{AwiszusWeber}, \cite{CifuentesFerrucciShin},  \cite{GaiKapadia}.
\end{itemize}

\noindent An axiomatic characterization of systemic risk measures of the form \eqref{112} on a finite state space is provided in \cite{ChenIyengarMoallemi}, see also \cite{KromerOverbeckZilch} for the extension to a general probability space and \cite{HoffmannMeyer-BrandisSvindland} for a further extension to a conditional setting. Also, in these references further examples of possible aggregation functions can be found.  Our framework may accommodate also such aggregation functions, provided these
satisfy the (simple) conditions outlined in Section \ref{one_dim_set}.

If $\eta$ in (\ref{112}) is a monetary risk measure it follows from (\ref{1}) that we can rewrite the systemic risk measure $\rho $ in (\ref{112}) as
\begin{equation}
\rho (\mathbf{X}):= \inf \{m\in \mathbb{R}\mid \Lambda (\mathbf{X}%
)+m\in \mathbb{A}\}\,.  \label{2}
\end{equation}
Thus, presuming $\Lambda (\mathbf{X})$ represents some loss, systemic risk can again be interpreted as the minimal cash amount that secures the system when it is added to the total aggregated system loss $\Lambda (\mathbf{X})$. If $\Lambda (\mathbf{X})$ does not allow for an interpretation as cash, the risk measure in \eqref{2} has to be understood as some general risk level of the system rather than some capital requirement. Similarly, if $\eta $ is a quasi-convex risk measure the systemic risk measure $\rho $ in (\ref{112}) can be rewritten as%
\begin{equation}\label{QCFA}
\rho (\mathbf{X}):= \inf \{m\in \mathbb{R}\mid \Lambda (\mathbf{X}%
)\in \mathbb{A}^{m}\}.
\end{equation}%
Again one first aggregates the risk factors via the function $\Lambda $ and
in a second step one computes the minimal risk level associated to $\Lambda (\mathbf{X})$. \newline
While the approach prescribed in \eqref{2} and \eqref{QCFA} defines an interesting class of systemic risk measures, one could think of meaningful alternative or extended procedures of measuring systemic risk not captured by  \eqref{2} or \eqref{QCFA}. The purpose of this paper is to specify a general methodological framework that is flexible enough to cover a wide range of possibilities to design systemic risk measures via acceptance sets and aggregation functions and to study corresponding examples. In the following subsections we extend the conceptual framework for systemic risk measures via acceptance sets step by step in order to gradually include certain novel key features of our approach.






\subsection{First add capital, then aggregate} 
The interpretation of \eqref{2} to measure systemic risk as minimal capital needed to secure the system \textit{after aggregating individual risks} is for example meaningful in the situation where some kind of rescue fund shall be installed to repair damage from systemic loss. However, for instance from the viewpoint of a regulator that has the possibility to intervene on the level of the single institutions before contagion effects generate further losses it might be more relevant to measure systemic risk as the minimal capital that secures the aggregated system by injecting the capital into the single institutions \textit{before aggregating the individual risks}. This way of measuring
systemic risk can be expressed by
\begin{equation} 
\rho (\mathbf{X}):= \inf \{\sum_{i=1}^{N}m_{i}\mid
\mathbf{m}=(m_{1},...,m_{N})\in \mathbb{R}^{N},\,\Lambda (\mathbf{X}+\mathbf{%
m})\in \mathbb{A}\}\,.  \label{3}
\end{equation}
Here, the amount $m_i$ is added to the financial position $X^i$ of institution $i\in\{1,...,N\}$ before the corresponding total loss $\Lambda (\mathbf{X}+\mathbf{m})$ is computed. For example, considering the aggregation function $\Lambda_{CM}$ from above it becomes clear that injecting cash first might prevent further losses that would be generated by contagion effects. The systemic risk is then measured as the minimal total amountl $\sum_{i=1}^{N}m_{i}$ injected into the institutions to secure the system.
\footnote{Independently a related concept in the context of set-valued systemic risk measures has been  developed in \cite{FeinsteinRudloffWeber}.}\\
Another interesting feature of the approach in \eqref{3} is that it  delivers at the same time a measure of total systemic risk as well as a potential ranking of the institutions in terms of systemic riskiness. Indeed, for $\mathbf{X}$ given, let $\mathbf{m^*}=(m^*_{1},...,m^*_{N})$ be such that $\rho (\mathbf{X})=\sum_{i=1}^{N}m^*_{i}$ and denote the ordered cash allocations by $m^*_{i_1}\ge...\ge m^*_{i_N}$. Then, one could argue that the risk factor $X^{i_1}$ that requires the biggest cash allocation $m^*_{i_1}$ corresponds to the systemic riskiest institution, $X^{i_2}$ corresponds to the systemic second riskiest institution, and so on. Of course, such allocation $\mathbf{m^*}$ does not need to be unique, in which case one has to discuss criteria that justify the choice of a specific allocation.

\subsection{First add scenario-dependent allocation, then aggregate} 
One
main novelty of this paper is that we want to allow for the possibility
of adding to $\mathbf{X}$ not merely a vector $\mathbf{m}=(m_{1},...,m_{N})%
\in \mathbb{R}^{N}$ of cash but a random vector
$$\mathbf{Y}\in \mathcal{C}\subseteq \mathcal{L}^{0}(\mathbb{R}^{N})$$
which represents admissible
assets with possibly random payoffs at time $T$, in the spirit of \cite{FrittelliScandolo}. To each $\mathbf{Y}\in \mathcal{C}$ we
assign a measure $\pi (\mathbf{Y})$ of the risk (or cost) associated to $%
\mathbf{Y}$ determined by a monotone increasing map
\begin{equation}\label{CashSet}
\pi :\mathcal{C}\rightarrow \mathbb{R}\,.
\end{equation}%
This leads to the following extension of \eqref{3}:
\begin{equation}
\rho (\mathbf{X}):= \inf \{\pi (\mathbf{Y})\in \mathbb{R}\mid
\mathbf{Y}\in \mathcal{C}, \,\Lambda (\mathbf{X}+\mathbf{Y})\in \mathbb{A}\}\,.
\label{4}
\end{equation}
Note that in order to establish a ranking of the institutions in a system $\mathbf{X}$ in terms of systemic riskiness implied by a $\mathbf{Y^*}=(Y^*_1,...,Y^*_N)\in \mathcal{C}$ with $\rho (\mathbf{X})=\pi (\mathbf{Y^*})$ in analogy to the ranking process described above implied by a deterministic $\mathbf{m^*}=(m^*_{1},...,m^*_{N})$ one now first has to introduce an ordering of the $Y^*_1,...,Y^*_N$. For example, one could say $X^i$ is systemic riskier than $X^j$ if $E[Y^*_i]>E[Y^*_j]$, presumed the expectations $E[Y^*_i]$, $i=1,...,N$ are well defined. \newline
Considering a general set $\mathcal{C}$ in \eqref{4} allows for more general
measurement of systemic risk than the cash needed today for each institution to secure the 
system. For example, $\mathcal{C}$ could be a set of (vectors of) general
admissible financial assets that can be
used to secure a system by adding $\mathbf{Y}$ to $\mathbf{X}$
component-wise, and $\pi (\mathbf{Y})$ is a valuation of $\mathbf{Y}$.
Another example that we focus on in this paper and which is particularly interesting from the viewpoint of a lender of last resort is the following class of sets $\mathcal{C}$:
\begin{equation}\label{def:RandCap}
\mathcal{C}\subseteq \{\mathbf{Y}\in \mathcal{L}^{0}(\mathbb{R}^{N})\mid
\sum_{n=1}^{N}Y^{n}\in \mathbb{R}\}=:\mathcal{C}_{\mathbb{R}},
\end{equation}%
and $\pi (\mathbf{Y})=\sum_{n=1}^{N}Y^{n}$.  Here the notation $\sum_{n=1}^{N}Y^{n}\in \mathbb{R}$ means that $\sum_{n=1}^{N}Y^{n}$ is equal to some deterministic constant in $\mathbb{R}$, even though each single $Y^n$, $n=1,\cdots,N$, is a random variable. Then, as in \eqref{3} the systemic risk measure
\begin{equation} \label{CashRM}
\rho (\mathbf{X}):= \inf \{\sum_{n=1}^{N}Y^{n}\mid \mathbf{Y}\in
\mathcal{C}, \,\Lambda (\mathbf{X}+\mathbf{Y})\in \mathbb{A}\}
\end{equation}%
can still be interpreted as the minimal total cash amount $\sum_{n=1}^{N}Y^{n}\in \mathbb{R}$ needed today to secure the system by
distributing the cash at the future time $T$ among the components of the
risk vector $\mathbf{X}$. However, contrary to \eqref{3}, in general the allocation $Y^i(\omega)$ to institution $i$ does not need to be decided today but depends on the scenario $\omega$ that has been realized at time $T$. This corresponds to the situation of a lender of last resort who is equipped with a certain amount of cash today and who will allocate it according to where it serves the most depending on the scenario that has been realized. Restrictions on the possible distributions of cash are given by the set $\mathcal{C}$. For example, for $\mathcal{C}=\mathbb{R}^{N}$ the situation corresponds to \eqref{3} where the distribution is already determined today, while for $\mathcal{C}=\mathcal{C}_{\mathbb{R}}$ the distribution can be chosen completely freely depending on the scenario $\omega$ that has been realized (including negative amounts, i.e. withdrawals of cash from certain components). \newline
Section \ref{sec:RandCash}, \ref{sec:Gaussian}, and \ref{finite_space} will be devoted to the analysis and concrete examples of the class of systemic risk measures using a set $\mathcal{C}$ as in \eqref{def:RandCap}. We will see that in the case $\mathcal{C}=\mathcal{C}_{\mathbb{R}}$ where unrestricted cross-subsidization is possible the canonical way of measuring systemic risk measure is of the form \eqref{112} with aggregation rule $\Lambda (\mathbf{x})=\sum_{i=1}^{N}x_{i},\, \mathbf{x}\in\mathbb{R}^{N}$, i.e. to apply a univariate risk measure to the sum of the risk factors. Another interesting feature of allowing scenario depending allocations of cash $\mathbf{Y} \in \mathcal{C} \subseteq \mathcal{C}_{\mathbb{R}}$ is that in general the systemic risk measure will take the dependence structure of the components of  $\mathbf{X}$ into account even though acceptable positions might be defined in terms of the marginal distributions of $X^i$, $i=1,...,N$ only. For instance, the example in Section \ref{sec:Gaussian} employs the aggregation rule $\Lambda (\mathbf{x})=\sum_{i=1}^{N}-x_{i}^{-},\, \mathbf{x}\in\mathbb{R}^{N}$, and the acceptance set $\mathbb{A}_{\gamma}:=\{Z\in\mathcal{L}^{0}(\mathbb{R})\mid E[Z]\ge\gamma\}\,,\gamma\in\mathbb{R}$. Then a risk vector $\mathbf{Z}=(Z_1,...,Z_N) \in\mathcal{L}^{0}(\mathbb{R}^{N})$ is acceptable if and only if $\Lambda (\mathbf{Z})\in\mathbb{A}$, i.e.
$$
\sum_{i=1}^{N}-E[Z_{i}^{-}]\ge\gamma\,,
$$
which only depends on the marginal distributions of $\mathbf{Z}$. Thus, if we choose $\mathcal{C}=\mathbb{R}^{N}$ then it is obvious that in this case also the systemic risk measure $\rho(\mathbf{X})$ in \eqref{CashRM} depends on the marginal distributions of $\mathbf{X}$ only. If, however, one allows for more general allocations $\mathbf{Y} \in \mathcal{C} \subseteq \mathcal{C}_{\mathbb{R}}$ that might differ from scenario to scenario the systemic risk measure will in general depend on the multivariate distribution of $\mathbf{X}$ since it can play on the dependence of the components of $\mathbf{X}$ to minimize the costs.

\subsection{Multi-dimensional Acceptance Sets}
Until now we have always defined systemic risk measures in terms of acceptability of an aggregated, one-dimensional loss figure. However, not necessarily every relevant systemic risk measure is of this aggregated type. Consider for instance the popular approach (though possibly problematic for financial systems as explained above) to add single univariate monetary risk measures $\eta _{i},\, i=1,...,N$, i.e.
\begin{equation}\label{SRMAdd}
\rho (\mathbf{X}):= \sum_{i=1}^N \eta _{i}(X^{i})\,.
\end{equation}%
In general, the systemic risk measure in \eqref{SRMAdd} cannot be expressed in the
form (\ref{4}). Denoting by $\mathbb{A}_{i}\subseteq \mathcal{L}^{0}(\mathbb{R})$ the acceptance set of $\eta _{i},\, i=1,...,N$, one easily sees from \eqref{1}, however, that $\rho$ in \eqref{SRMAdd} can be written in terms of the multivariate acceptance set $\mathbb{A}_{1}\times ...\times \mathbb{A}_{N}$:
\begin{equation*}
\rho (\mathbf{X}):= \inf \{\sum_{i=1}^N m_{i}\mid \,\mathbf{m}=(m_1,...,m_N)\in \mathbb{R}^N,\ \mathbf{X}+
\mathbf{m}\in \mathbb{A}_{1}\times ...\times \mathbb{A}_{N}\}\,.
\end{equation*}
Motivated by this example, we extend (\ref{4}) to the formulation of systemic risk measures as the minimal cost of admissible asset vectors $\mathbf{Y}\in \mathcal{C}$ that, when added to the vector of financial positions $\mathbf{X}$, makes the augmented financial positions $\mathbf{X}+\mathbf{Y}$ acceptable in terms of a general multidimensional acceptance set $\mathcal{A}\subseteq \mathcal{L}^{0}(\mathbb{R}^{N})$:
\begin{equation}
\rho (\mathbf{X}):= \inf \{\pi (\mathbf{Y})\in \mathbb{R}\mid
\mathbf{Y}\in \mathcal{C},\,\mathbf{X}+\mathbf{Y}\in \mathcal{A}\}\,.
\label{5}
\end{equation}
Note that by putting $\mathcal{A}:=\left\{\mathbf{Z}\in \mathcal{L}^{0}(\mathbb{R}^{N})\mid \Lambda(\mathbf{Z})\in \mathbb{A}\right\}$ Definition (\ref{4}) is a special case of \eqref{5}. Also, in analogy to \eqref{CAdd}, we remark that for linear valuation rules $\pi $ the systemic risk measure given in \eqref{5} exhibits an extended type of cash invariance in the sense that
\begin{equation} \label{CashAdd}
\rho (\mathbf{X}+\mathbf{Y})=\rho(\mathbf{X})+\pi (\mathbf{Y})
\end{equation}%
for $\mathbf{Y}\in \mathcal{C}$ such that $\mathbf{Y^{\prime }}\pm \mathbf{Y}%
\in \mathcal{C}$ for all $\mathbf{Y^{\prime }}\in \mathcal{C}$, see \cite{FrittelliScandolo}.

\subsection{Degree of Acceptability}
In order to reach the final, most general formulation of systemic risk measures, we assign, in analogy to \eqref{11}, to each $\mathbf{Y}\in \mathcal{C}$ a set $\mathcal{A}^{\mathbf{Y}}\subseteq \mathcal{L}^{0}(\mathbb{%
R}^{N})$ of risk vectors that are acceptable for the given
(random) vector $\mathbf{Y}$, and define the systemic risk measure by:
\begin{equation}
\rho (\mathbf{X}):= \inf \{\pi (\mathbf{Y})\in \mathbb{R}\mid
\mathbf{Y}\in \mathcal{C},\,\mathbf{X}\in \mathcal{A}^{\mathbf{Y}}\}\,.
\label{6}
\end{equation}
Note that analogously to the one-dimensional quasi-convex case \eqref{11}, the systemic risk measures \eqref{6} cannot necessarily be interpreted as cash added to the system but in general represents some minimal aggregated risk level $\pi (\mathbf{Y})$ at which the system $\mathbf{X}$ is acceptable. The approach in (\ref{6}) is very flexible and unifies a variety of different features in the
design of systemic risk measures. In particular, it includes all previous cases if we set
\begin{equation*}
\mathcal{A}^{\mathbf{Y}}:=\mathcal{A}-\mathbf{Y},
\end{equation*}%
where the set $\mathcal{A}\subseteq \mathcal{L}^{0}(\mathbb{R}^{N})$
represents acceptable risk vectors. Then obviously (\ref{5}) is obtained from (\ref{6}).

Another advantage of the formulation in terms of general acceptance
sets is the possibility to design systemic risk measures via general aggregation rules. Indeed the formulation  \eqref{6} includes the case
\begin{equation} \label{GenAggr}
\rho (\mathbf{X}):= \inf \{\pi (\mathbf{Y})\in \mathbb{R}\mid
\mathbf{Y}\in \mathcal{C},\,\Theta (\mathbf{X},\mathbf{Y)}\in \mathbb{A}\},
\end{equation}%
where $\Theta :\mathcal{L}^{0}(\mathbb{R}^{N})\times \mathcal{C}\rightarrow
\mathcal{L}^{0}(\mathbb{R})$ denotes some aggregation function  jointly in $\mathbf{X}$ and $\mathbf{Y}$. Just select $\mathcal{A}^{\mathbf{Y}}:=\left\{ \mathbf{Z}\in \mathcal{L}^{0}(\mathbb{R}^{N})\mid \Theta (\mathbf{Z},\mathbf{%
Y)}\in \mathbb{A}\right\} .$ In particular, \eqref{GenAggr} includes both the case  \textquotedblleft injecting capital before aggregation\textquotedblright  as in (\ref{3})  and \eqref{4} by putting
$\Theta (\mathbf{X},\mathbf{Y)=}\Lambda (\mathbf{X}+\mathbf{Y})$, and the case \textquotedblleft  aggregation before injecting capital\textquotedblright  as in (\ref{2}) by putting $\Theta (\mathbf{X},\mathbf{Y):=}%
\Lambda _{1}(\mathbf{X)+}\Lambda _{2}(\mathbf{Y),}$ where $\Lambda _{1}:%
\mathcal{L}^{0}(\mathbb{R}^{N})\rightarrow \mathcal{L}^{0}(\mathbb{R})$ is
an aggregation function and $\Lambda _{2}:\mathcal{C}\rightarrow \mathcal{L}%
^{0}(\mathbb{R})$ could be, for example, the discounted cost of $\mathbf{Y}$. \newline
Also, again in analogy to the one-dimensional case \eqref{11}, the more general dependence of
the acceptance set on $\mathbf{Y}$ in (\ref{6}) allows for multi-dimensional quasi-convex risk measures. Note that the cash additivity property \eqref{CashAdd} is then lost in general.\newline
The remainder of the paper is organized as follows. In the next Section we structure and lay the theoretical foundations of the approach to systemic risk measures motivated and outlined above. In particular, we provide reasonable conditions on the ingredients $\mathcal{C}$, $\pi $, $\mathcal{A}$, $\mathbb{A}$, and $\Lambda $ such that the above definitions of $\rho $ are well posed and $\rho $ has the natural properties of decreasing monotonicity
and quasi-convexity (or convexity). In Section \ref{sec:Aggr} we analyze the situation and give various families of systemic risk measures when the risk measurement is defined in terms of the natural approach to apply some kind of aggregation to risk factors and test acceptability with respect to some one-dimensional acceptance set as in \eqref{GenAggr} above. Section \ref{sec:RandCash} investigates the interesting class of systemic risk measures that are defined in terms of a set $\mathcal{C}$ of scenario-dependent allocations as in \eqref{def:RandCap}. Then we present two concrete examples within this class of systemic risk measures in Sections \ref{sec:Gaussian} and \ref{finite_space}. In Section \ref{sec:Gaussian} we look Gaussian systems and consider both deterministic cash allocations as well as a certain class of random cash allocations. Further, we apply the results to a particular Gaussian system where the flow of money between the institutions (borrowing and lending) is modeled by a system of interacting diffusions (see \cite{CarmonaFouqueSun}). In Section \ref{finite_space} we introduce an example on a finite probability space. As a consequence of the finite probability space we are able to compute systemic risk measures for very general random cash allocations $\mathcal{C} \subseteq \mathcal{C}_{\mathbb{R}}$.


\section{Definition of Systemic Risk Measures and Properties\label{sec:GenSetting}}

In this Section we provide the definitions and properties of the systemic risk measures in our setting. As in the Introduction, we
consider the set of random vectors 
\begin{equation*}
\mathcal{L}^{0}(\mathbb{R}^{N}):=\{\mathbf{X}=(X^{1},\ldots ,X^{N})\mid X^{n}\in
\mathcal{L}^{0}(\Omega ,\mathcal{F}, \mathbb{P}),\;n=1,\cdots,N\},
\end{equation*}
on the probability space $(\Omega ,\mathcal{F}, \mathbb{P})$. 
 We assume that $%
\mathcal{L}^{0}(\mathbb{R}^{N})$ is equipped with an order relation $\succeq $ such that it is a
vector lattice. One such example is provided by the order relation: $\mathbf{X}%
_{1}\succeq \mathbf{X}_{2}$ if $X_{1}^{i}\geq X_{2}^{i}$ for all components $%
i=1,...,N$, where for random variables on $\mathcal{L}^{0}(\mathbb{R})$, the order relation is determined by $\mathbb{P}-$a.s
inequality.

\begin{definition}
Let $\mathbf{X}_{1},$ $\mathbf{X}_{2}\in \mathcal{L}^{0}(\mathbb{R}^{N})$.
\begin{enumerate}
\item A set $\mathcal{A}\subset \mathcal{L}^{0}(\mathbb{R}^{N})$ is $\succeq $
-monotone if $\mathbf{X}_{1}\in \mathcal{A}$ and $\mathbf{X}_{2}\succeq 
\mathbf{X}_{1}$ implies $\mathbf{X}_{2}\in \mathcal{A}$. 
\item A map $f:\mathcal{L}^{0}(\mathbb{R}^{N})\rightarrow \mathcal{L}^{0}(\mathbb{R
})$\ is $\succeq $-monotone decreasing if $\mathbf{X}_{2}\succeq \mathbf{X}
_{1}$ implies $f(\mathbf{X}_{1})\geq f(\mathbf{X}_{2}).$ Analogously for
functions $f:\mathcal{L}^{0}(\mathbb{R}^{N})\rightarrow \overline{\mathbb{R}}
$. 
\item A map $f:\mathcal{L}^{0}(\mathbb{R}^{N})\rightarrow \overline{\mathbb{R}}$\
is quasi-convex if 
\begin{equation*}
f(\lambda \mathbf{X}_{1}+(1-\lambda )\mathbf{X}_{2})\leq f(\mathbf{X}%
_{1})\vee f(\mathbf{X}_{2}).
\end{equation*}
\end{enumerate}
\end{definition}
A vector $\mathbf{X}=(X^{1},\ldots ,X^{N})\in \mathcal{L}^{0}(\mathbb{R}^{N})
$ denotes a configuration of risky factors at a future time $T$ associated
to a system of $N$ entities. Let 
\begin{equation*}
\mathcal{C}\subseteq \mathcal{L}^{0}(\mathbb{R}^{N}).
\end{equation*}%
To each $\mathbf{Y}\in \mathcal{C}$ we assign a set $\mathcal{A}^{%
\mathbf{Y}}\subseteq \mathcal{L}^{0}(\mathbb{R}^{N})$. The set $\mathcal{A}^{%
\mathbf{Y}}$ represents the  risk vectors $
\mathbf{X}$ that are acceptable for the given random vector $\mathbf{Y}$. Let also consider 
a map 
\begin{equation*}
\pi :\mathcal{C}\rightarrow \mathbb{R}\,,
\end{equation*}%
so that $\pi (\mathbf{Y})$ represents the risk (or cost) associated to $\mathbf{Y}$. \newline
We now introduce the concept of monotone and (quasi-) convex systemic risk measure.
\begin{definition}
The \textit{systemic risk measure} associated with $ \mathcal{C}, \mathcal{A}^{
\mathbf{Y}}$ and $\pi$ is  a map $\rho :\mathcal{L}^{0}(\mathbb{R}%
^{N})\rightarrow \overline{\mathbb{R}}:=\mathbb{R}\cup \left\{ -\infty
\right\} \cup \left\{ \infty \right\} $, defined by: 
\begin{equation}
\rho (\mathbf{X}):= \inf \{\pi (\mathbf{Y})\in \mathbb{R}\mid 
\mathbf{Y}\in \mathcal{C},\,\mathbf{X}\in \mathcal{A}^{\mathbf{Y}}\}\,.
\label{00}
\end{equation}
Moreover $\rho$ is called a quasi-convex (resp. convex) systemic risk measure
if it is $\succeq $-monotone decreasing and quasi-convex (resp. convex on $
\left\{ \rho(\mathbf{X)<+\infty }\right\} $).
\end{definition}
In other words, the systemic risk of a random vector $\mathbf{X}$ is measured by the
minimal risk (cost) of those random vectors $\mathbf{Y}$ that make $\mathbf{X
}$ acceptable. \\
As already sketched in the Introduction, we now focus on several examples of
systemic risk measures of the type (\ref{00}). To guarantee that such maps
are finite valued one could consider their restriction to some vector
subspaces of $\mathcal{L}^{0}(\mathbb{R}^{N})$ (for examples $\mathcal{L}
^{p}(\mathbb{R}^{N}),$ $p\in \lbrack 1,\infty ]$) and impose further
conditions on the defining ingredients ($\pi ,$ $\mathcal{C}$, $\mathcal{A}^{
\mathbf{Y}}$) of $\rho $. For example, suppose that $\mathcal{C}$  and $\mathcal{A}^{
\mathbf{Y}}$ satisfy the two conditions
\begin{align*}
&\left\{ m\mathbf{1\in }\mathbb{R}\mathbf{^{N}\mid }\text{ }m\in 
\mathbb{R}_{+}\text{, }\mathbf{1:=(}1\mathbf{,...,}1\mathbf{)}\right\}
\subseteq \mathcal{C}, \\
&-m\mathbf{1}\in \mathcal{A}^{m\mathbf{1}} \textrm{\ and \ } \mathcal{A}^{m\mathbf{
1}} \textrm{\ is a monotone set for each \ } m\in \mathbb{R}_{+},
\end{align*}
then $\rho :\mathcal{L}^{\infty }(\mathbb{R}^{N})\rightarrow 
\overline{\mathbb{R}}$ defined by (\ref{00}) satisfies $\rho (\mathbf{X}%
)<+\infty $ for all $\mathbf{X}\in \mathcal{L}^{\infty }(\mathbb{R}^{N})$.
Indeed, for $m:=\max_{i}\|X^{i}\| _{\infty }$, $\mathbf{X}%
\geq -m\mathbf{1}\in \mathcal{A}^{m\mathbf{1}}$ implies that $\mathbf{X\in }
\mathcal{A}^{m\mathbf{1}}$ and $\pi (m\mathbf{1})<+\infty$. 
\newline
Clearly, other sufficient conditions may be obtained in each specific
example of systemic risk measures considered in the subsequent Sections. \newline
We opt to accept the possibility that such maps $\rho $ may assume values $
\pm \infty $. However, it is not difficult to find simple sufficient
conditions assuring that the systemic risk measure in (\ref{00}) is proper
(not identically equal to $+\infty )$. One such example is the condition:

\

\centerline{
if $\mathbf{0}\in \mathcal{C}$ and $\mathbf{0}\in \mathcal{A}^{\mathbf{
0}}$ then $\rho (\mathbf{0})\leq \pi (\mathbf{0})<+\infty .$}

\

\noindent We now consider the \textquotedblleft structural
properties\textquotedblright\ (i.e. monotonicity, quasiconvexity, convexity)
of our systemic risk measures and introduce two sets of conditions
(properties  \ref{condmon}, \ref{condqc} and \ref{condc} below 
and the alternative properties \ref{condqc2} and  \ref{condc2}) that 
guarantee that the map in (\ref{00}) is a quasi-convex (or convex) risk
measure. In Section \ref{sec:Aggr} we show that these sets of
conditions can be easily checked in some relevant examples of maps in the
form (\ref{00}), where the set $\mathcal{A}^{\mathbf{Y}}$ is determined from
aggregation and one-dimensional acceptance sets. \newline
We  introduce the following properties:
\begin{description}
\item[(P1)\label{condmon}]For all $\mathbf{Y}\in \mathcal{C}$ the set $%
\mathcal{A}^{\mathbf{Y}}\subset \mathcal{L}^{0}(\mathbb{R}^{N})$ is $\succeq 
$-monotone.

\item[(P2)\label{condqc}] For all $m\in \mathbb{R}$, for all $\mathbf{Y}_{1},%
\mathbf{Y}_{2}\in \mathcal{C}$ such that $\pi (\mathbf{Y}_{1})\leq m$ and $%
\pi (\mathbf{Y}_{2})\leq m$ and for all $\mathbf{X}_{1}\in \mathcal{A}^{%
\mathbf{Y}_{1}}$, $\ \mathbf{X}_{2}\in \mathcal{A}^{\mathbf{Y}_{2}}$ and all 
$\lambda \in \lbrack 0,1]$ there exists $\mathbf{Y}\in \mathcal{C}$ such
that $\pi (\mathbf{Y})\leq m$ and $\lambda \mathbf{X}_{1}+(1-\lambda )%
\mathbf{X}_{2}\in \mathcal{A}^{\mathbf{Y}}$.

\item[(P3)\label{condc}] For all $\mathbf{Y}_{1},\mathbf{Y}_{2}\in\mathcal{C}$
and all $\mathbf{X}_{1}\in \mathcal{A}^{\mathbf{Y}_{1}}$, $\ \mathbf{X}%
_{2}\in \mathcal{A}^{\mathbf{Y}_{2}}$ and all $\lambda \in \lbrack 0,1]$
there exists $\mathbf{Y}\in \mathcal{C}$ such that $\pi (\mathbf{Y})\leq
\lambda \pi(\mathbf{Y_1}) + (1-\lambda)\pi(\mathbf{Y_2})$ and $\lambda 
\mathbf{X}_{1}+(1-\lambda )\mathbf{X}_{2}\in \mathcal{A}^{\mathbf{Y}}$.
\end{description}
It is clear that property \ref{condc} implies property \ref{condqc}.
Moreover, we have:

\begin{lemma}
\label{lem:rm}

\begin{enumerate}

\item[i)] If  the systemic risk measure  $\rho $
defined in (\ref{00}) satisfies  the properties \ref{condmon} and \ref{condqc}, then $\rho$ is $\succeq $-monotone decreasing and quasi-convex.

\item[ii)]If  the systemic risk measure  $\rho $
defined in (\ref{00}) satisfies  the properties  \ref{condmon} and \ref{condc}, then $\rho $
 is $\succeq $-monotone decreasing and convex on $
\left\{ \rho (\mathbf{X)<+\infty }\right\} $.
\end{enumerate}
\end{lemma}
\begin{proof}
Set
\begin{equation*}
B(\mathbf{X}):=\left\{ \mathbf{Y}\in C\mid \mathbf{X}\in \mathcal{A}^{%
\mathbf{Y}}\right\} .
\end{equation*}%
First assume that property \ref{condmon} holds and w.l.o.g. suppose $\mathbf{%
X}_{2}\succeq \mathbf{X}_{1}$ and $B(\mathbf{X}_{1})\neq \varnothing $. Then
property \ref{condmon} implies that if $\mathbf{X}_{1}\in \mathcal{A}^{Y}$
and $\mathbf{X}_{2}\succeq \mathbf{X}_{1}$ then: $B(\mathbf{X}_{1})\subseteq
B(\mathbf{X}_{2})$. Hence:%
\begin{equation*}
\rho (\mathbf{X}_{1})=\inf \{\pi (\mathbf{Y})\mid \mathbf{Y}\in B(\mathbf{X}%
_{1})\}\geq \inf \{\pi (\mathbf{Y})\mid \mathbf{Y}\in B(\mathbf{X}%
_{2})\}=\rho (\mathbf{X}_{2})
\end{equation*}%
so that $\rho $ is $\succeq $-monotone decreasing.
\begin{enumerate}
\item[i)] Now assume that property \ref{condqc} holds and let $\mathbf{X}_{1},%
\mathbf{X}_{2}\in \mathcal{L}^{0}(\mathbb{R}^{N})$ be arbitrarily chosen.
For the quasi-convexity we need to prove, for any $m\in \mathbb{R}$, that:%
\begin{equation*}
\rho (\mathbf{X}_{1})\leq m\text{ and }\rho (\mathbf{X}_{2})\leq
m\Rightarrow \rho (\lambda \mathbf{X}_{1}+(1-\lambda )\mathbf{X}_{2})\leq m.
\end{equation*}%
By definition of the infimum in the definition of $\rho (\mathbf{X}_{i})$, $%
\forall \varepsilon >0$ there exist $\mathbf{Y}_{i}\in \mathcal{C}$ such
that $\mathbf{X}_{i}\in \mathcal{A}^{\mathbf{Y}_{i}}$ and%
\begin{equation*}
\pi (\mathbf{Y}_{i})\leq \rho (\mathbf{X}_{i})+\varepsilon \leq
m+\varepsilon, \ i=1,2.
\end{equation*}%
Take any $\lambda \in \lbrack 0,1]$. Property \ref{condqc} guarantees the
existence of $\mathbf{Z}\in \mathcal{C}$ such that $\pi (\mathbf{Z})\leq
m+\varepsilon $ and $\lambda \mathbf{X}_{1}+(1-\lambda )\mathbf{X}_{2}\in $ $%
\mathcal{A}^{\mathbf{Z}}$. Hence 
\begin{eqnarray*}
\rho (\lambda \mathbf{X}_{1}+(1-\lambda )\mathbf{X}_{2}) &=&\inf \{\pi (%
\mathbf{Y})\mid \mathbf{Y}\in \mathcal{C},\,\lambda \mathbf{X}%
_{1}+(1-\lambda )\mathbf{X}_{2}\in \mathcal{A}^{\mathbf{Y}}\} \\
&\leq &\pi (\mathbf{Z})\leq m+\varepsilon .
\end{eqnarray*}%
As this holds for any $\varepsilon >0$, we obtain the quasi-convexity.

\item[ii)]Assume that property \ref{condc} holds and that $\mathbf{X}_{1},\mathbf{X%
}_{2}\in \mathcal{L}^{0}(\mathbb{R}^{N})$ satisfy $\rho (\mathbf{X}_{i}%
\mathbf{)<+\infty }$. Then $B(\mathbf{X}_{i})\neq \varnothing $ and, as
before, $\forall \varepsilon >0$ there exists $\mathbf{Y}_{i}\in \mathcal{L}%
^{0}(\mathbb{R}^{N})$ such that: $\mathbf{Y}_{i}\in \mathcal{C}$, $\mathbf{X}%
_{i}\in \mathcal{A}^{\mathbf{Y}_{i}}$ and%
\begin{equation}
\pi (\mathbf{Y}_{i})\leq \rho (\mathbf{X}_{i})+\varepsilon \text{, }i=1,2.
\label{pirho}
\end{equation}%
By property \ref{condc} there exists $\mathbf{Z}\in \mathcal{C}$ such that $%
\pi (\mathbf{Z})\leq \lambda \pi (\mathbf{Y_{1}})+(1-\lambda )\pi (\mathbf{%
Y_{2}})$ and $\lambda \mathbf{X}_{1}+(1-\lambda )\mathbf{X}_{2}\in $ $%
\mathcal{A}^{\mathbf{Z}}$. Hence 
\begin{eqnarray*}
\rho (\lambda \mathbf{X}_{1}+(1-\lambda )\mathbf{X}_{2}) &=&\inf \{\pi (%
\mathbf{Y})\mid \mathbf{Y}\in \mathcal{C},\,\lambda \mathbf{X}%
_{1}+(1-\lambda )\mathbf{X}_{2}\in \mathcal{A}^{\mathbf{Y}}\} \\
&\leq &\pi (\mathbf{Z})\leq \lambda \pi (\mathbf{Y}_{1})+(1-\lambda )\pi (%
\mathbf{Y}_{2}) \\
&\leq &\lambda \rho (\mathbf{X}_{1})+(1-\lambda )\rho (\mathbf{X}%
_{2})+\epsilon ,
\end{eqnarray*}%
from (\ref{pirho}). As this holds for any $\varepsilon >0$, the map $\rho $
is convex on $\left\{ \rho (\mathbf{X)<+\infty }\right\}$.
\end{enumerate}
\end{proof}

\noindent We now consider the following alternative properties:

\begin{description}
\item[(P2a)\label{condqc2}] For all $\mathbf{Y}_{1},\mathbf{Y}_{2}\in \mathcal{C%
},$ $\mathbf{X}_{1}\in \mathcal{A}^{\mathbf{Y}_{1}}$, $\ \mathbf{X}_{2}\in 
\mathcal{A}^{\mathbf{Y}_{2}}\ $and $\lambda \in \lbrack 0,1]$ there exists $%
\alpha \in \lbrack 0,1]$ such that $\lambda \mathbf{X}_{1}+(1-\lambda )%
\mathbf{X}_{2}\in $ $\mathcal{A}^{\alpha \mathbf{Y}_{1}+(1-\alpha )\mathbf{Y}%
_{2}}$.

\item[(P3a)\label{condc2} ] For all $\mathbf{Y}_{1},\mathbf{Y}_{2}\in \mathcal{C%
},$ $\mathbf{X}_{1}\in \mathcal{A}^{\mathbf{Y}_{1}}$, $\ \mathbf{X}_{2}\in 
\mathcal{A}^{\mathbf{Y}_{2}}\ $and $\lambda \in \lbrack 0,1]$ it holds: $%
\lambda \mathbf{X}_{1}+(1-\lambda )\mathbf{X}_{2}\in $ $\mathcal{A}^{\lambda 
\mathbf{Y}_{1}+(1-\lambda )\mathbf{Y}_{2}}$.
\end{description}

\noindent It is clear that property \ref{condc2} implies  property \ref{condqc2}.
Furthermore we introduce the following properties for 
$\mathcal{C}$ and $\pi$:

\begin{description}
\item[(P4)\label{condCc}] $\mathcal{C}$ is convex,


\item[(P5)\label{condpiqc}] $\pi $ is quasi-convex,

\item[(P6)\label{condpic}] $\pi $ is convex.
\end{description}

\noindent We have the following:

\begin{lemma}
\label{Lem2a3a}
\begin{enumerate}

\item[i)]  Under the conditions \ref{condmon}, \ref{condqc2}, \ref{condCc}, and \ref%
{condpiqc} the map $\rho $ defined in (\ref{00}) is a quasi-convex systemic risk
measure.

\item[ii)]  Under the conditions \ref{condmon}, \ref{condc2}, \ref{condCc} and \ref%
{condpic} the map $\rho $ defined in (\ref{00}) is a convex systemic risk measure.
\end{enumerate}
\end{lemma}
\begin{proof}

 \begin{enumerate}

\item[i)]  It follows from Lemma \ref{lem:rm} and the fact that the properties \ref%
{condqc2}, \ref{condCc}, and \ref{condpiqc} imply \ref{condqc}. Indeed, let $%
\mathbf{Y}_{1},\mathbf{Y}_{2}\in \mathcal{C}$ such that $\pi (\mathbf{Y}%
_{1})\leq m$, $\pi (\mathbf{Y}_{2})\leq m$ and let $\mathbf{X}_{1}\in 
\mathcal{A}^{\mathbf{Y}_{1}}$, $\ \mathbf{X}_{2}\in \mathcal{A}^{\mathbf{Y}%
_{2}}\ $ and $\lambda \in \lbrack 0,1]$. Then there exists $\alpha \in
\lbrack 0,1]$ such that $\lambda \mathbf{X}_{1}+(1-\lambda )\mathbf{X}%
_{2}\in $ $\mathcal{A}^{\alpha \mathbf{Y}_{1}+(1-\alpha )\mathbf{Y}_{2}}$.
If we set: $\mathbf{Y:=}\alpha \mathbf{Y}_{1}+(1-\alpha )\mathbf{Y}_{2}\in 
\mathcal{C}$ then $\lambda \mathbf{X}_{1}+(1-\lambda )\mathbf{X}_{2}\in $ $%
\mathcal{A}^{\mathbf{Y}}$ and $\pi (\alpha \mathbf{Y}_{1}+(1-\alpha )\mathbf{%
Y}_{2})\leq \max (\pi (\mathbf{Y}_{1}),\pi (\mathbf{Y}_{2}))\leq m$.

\item[ii)]  It follows from Lemma \ref{lem:rm} and the fact that the properties 
\ref{condc2}, \ref{condCc}, and \ref{condpic} imply \ref{condc}. Indeed, let 
$\mathbf{Y}_{1},\mathbf{Y}_{2}\in \mathcal{C}$ and let $\mathbf{X}_{1}\in 
\mathcal{A}^{\mathbf{Y}_{1}}$, $\ \mathbf{X}_{2}\in \mathcal{A}^{\mathbf{Y}%
_{2}}\ $ and $\lambda \in \lbrack 0,1]$. If we set: $\mathbf{Y:=}\lambda 
\mathbf{Y}_{1}+(1-\lambda )\mathbf{Y}_{2}\in \mathcal{C}$ then $\lambda 
\mathbf{X}_{1}+(1-\lambda )\mathbf{X}_{2}\in $ $\mathcal{A}^{\mathbf{Y}}$
and $\pi (\mathbf{Y})\leq \lambda \pi (\mathbf{Y_{1}})+(1-\lambda )\pi (%
\mathbf{Y_{2}})$.
\end{enumerate}
\end{proof}


\section{Systemic Risk Measures via Aggregation and One-dimensional
Acceptance Sets \label{one_dim_set}}

\label{sec:Aggr}

In this Section we study four classes of systemic risk measures in the form \eqref{00}, which differ from each other by the definition of their
aggregation functions and their acceptance sets. However, these four classes,
defined in equations (\ref{0011}), (\ref{0012}), (\ref{0013}), (\ref{0014}),
all satisfy the structural properties of monotonicity and quasi-convexity (or
convexity). We consider the following definitions and assumptions, which will hold true throughout this Section: 
\begin{enumerate}
\item the aggregation functions are
\begin{eqnarray*}
\Lambda &:&\mathcal{L}^{0}(\mathbb{R}^{N})\times \mathcal{C}\rightarrow 
\mathcal{L}^{0}(\mathbb{R}), \\
\Lambda _{1} &:&\mathcal{L}^{0}(\mathbb{R}^{N})\rightarrow \mathcal{L}^{0}(\mathbb{R}),
\end{eqnarray*}%
and we assume that $\Lambda _{1}$ is $\succeq $-increasing and concave; 
\item the acceptance family 
\begin{equation*}
(\mathcal{B}^{x})_{x\in \mathbb{R}}
\end{equation*}
is an increasing family with respect to $x$ and each set $\mathcal{B}^{x}\subseteq \mathcal{L}^{0}(\mathbb{%
R})$ is assumed monotone and convex;

\item the acceptance subset  
\begin{equation*}
\mathbb{A}\subseteq \mathcal{L}^{0}(\mathbb{R}).
\end{equation*}
is assumed monotone and convex.
\end{enumerate}
The convexity of the
acceptance set $\mathbb{A}\subseteq \mathcal{L}^{0}(\mathbb{R})$ (or of the
acceptance family $(\mathcal{B}^{x})_{x\in \mathbb{R}}$) are the standard
conditions which have been assumed since the origin of the theory of risk
measures. The concavity of the aggregation functions is justified, not only
from the many relevant examples in literature, but also by the preservation of the
convexity from one dimensional acceptance sets to multi-dimensional ones.
Indeed, let $\Theta :\mathcal{L}^{0}(\mathbb{R}^{N})\rightarrow \mathcal{L}%
^{0}(\mathbb{R})$ be an aggregation function$,$ $\mathbb{A}\subseteq 
\mathcal{L}^{0}(\mathbb{R})$ a one dimensional acceptance set and define $%
\mathcal{A}\subseteq \mathcal{L}^{0}(\mathbb{R}^{N})$ as the inverse image $%
\mathcal{A}:=\Theta ^{-1}(\mathbb{A})$. Suppose that $\Theta $ is increasing
and concave. Then one may easily check that if $\mathbb{A}$ is monotone and
convex, then $\mathcal{A}$ is monotone and convex. \\
We note that in all results of this Section, the selection of the set $%
\mathcal{C}\subseteq \mathcal{L}^{0}(\mathbb{R}^{N})$ of permitted vectors
is left as general as possible (in some cases we require the convexity of $%
\mathcal{C}$ and only in Proposition \ref{prop:qc3} we further ask that $%
\mathcal{C}+\mathbb{R}_{+}^{N}\in \mathcal{C}$). Therefore, we are very
flexible in the choice of $\mathcal{C}$ and we may interpret its elements as
vectors of admissible or safe financial assets, or merely as cash vectors.
Only in the next Section we attribute a particular structure to $%
\mathcal{C}$.\\
In the conclusive statements of the following propositions in this section we 
apply Lemma \ref{lem:rm} and Lemma \ref{Lem2a3a} without explicit mention.
\begin{proposition}
Let
\begin{equation}
\mathcal{A}^{\mathbf{Y}}:=\left\{ \mathbf{Z}\in \mathcal{L}^{0}(\mathbb{R}%
^{N})\mid \Lambda (\mathbf{Z},\mathbf{Y)}\in \mathbb{A}\right\} ,\text{\quad 
}\mathbf{Y\in }\mathcal{C},  \label{AA}
\end{equation}%
$\Lambda $ be concave, and $\Lambda (\mathbf{\cdot },\mathbf{Y)}$ be $\succeq 
$-increasing for all $\mathbf{Y}\in \mathcal{C}$. Then $\mathcal{A}^{\mathbf{%
Y}}$ satisfies properties \ref{condmon} and \ref{condc2} (and \ref{condqc2}%
). The map $\rho $ defined in (\ref{00}) is given by 
\begin{equation}
\rho (\mathbf{X}):= \inf \{\pi (\mathbf{Y})\in \mathbb{R}\mid 
\mathbf{Y}\in \mathcal{C},\,\Lambda (\mathbf{X},\mathbf{Y)}\in \mathbb{A}\}\,,
\label{0011}
\end{equation}%
and is a quasi-convex systemic risk measure, under the assumptions \ref{condCc} and %
\ref{condpiqc}; it is a convex systemic risk measure under the assumptions \ref{condCc} and \ref{condpic}.
\end{proposition}
\begin{proof}
Property \ref{condmon}: Let $\mathbf{X}_{1}\in \mathcal{A}^{\mathbf{Y}}$ and 
$\mathbf{X}_{2}\succeq \mathbf{X}_{1}$. Note that $\mathbf{X}_{1}\in 
\mathcal{A}^{\mathbf{Y}}$ implies $\Lambda (\mathbf{X}_{1},\mathbf{Y)}\in 
\mathbb{A}$ and $\mathbf{X}_{2}\succeq \mathbf{X}_{1}$ implies $\Lambda (%
\mathbf{X}_{2},\mathbf{Y)\geq }\Lambda (\mathbf{X}_{1},\mathbf{Y)}$. Since $%
\mathbb{A}$ is monotone we have $\Lambda (\mathbf{X}_{2},\mathbf{Y)}\in 
\mathbb{A}$\ and $\mathbf{X}_{2}\in \mathcal{A}^{\mathbf{Y}}.$\\
Property \ref{condc2}: Let $\mathbf{Y}_{1},\mathbf{Y}_{2}\in \mathcal{C},$ $%
\mathbf{X}_{1}\in \mathcal{A}^{\mathbf{Y}_{1}}$, $\ \mathbf{X}_{2}\in 
\mathcal{A}^{\mathbf{Y}_{2}}\ $and $\lambda \in \lbrack 0,1]$. Then $\Lambda
(\mathbf{X}_{1},\mathbf{Y}_{1}\mathbf{)}\in \mathbb{A}$ and $\Lambda (%
\mathbf{X}_{2},\mathbf{Y}_{2}\mathbf{)}\in \mathbb{A}$ and the convexity of $%
\mathbb{A}$ guarantees: 
\begin{equation*}
\lambda \Lambda (\mathbf{X}_{1},\mathbf{Y}_{1}\mathbf{)+(}1-\lambda \mathbf{)%
}\Lambda (\mathbf{X}_{2},\mathbf{Y}_{2}\mathbf{)}\in \mathbb{A}.
\end{equation*}%
From the concavity of $\Lambda (\mathbf{\cdot },\mathbf{\cdot )}$ we obtain:%
\begin{equation*}
\Lambda (\lambda (\mathbf{X}_{1},\mathbf{Y}_{1}\mathbf{)+(}1-\lambda \mathbf{%
)}(\mathbf{X}_{2},\mathbf{Y}_{2}\mathbf{)})\geq \lambda \Lambda (\mathbf{X}%
_{1},\mathbf{Y}_{1}\mathbf{)+(}1-\lambda \mathbf{)}\Lambda (\mathbf{X}_{2},%
\mathbf{Y}_{2}\mathbf{)}\in \mathbb{A}.
\end{equation*}%
The monotonicity of $\mathbb{A}$ implies: 
\begin{equation*}
\Lambda (\lambda \mathbf{X}_{1}\mathbf{+(}1-\lambda \mathbf{)X}_{2},\lambda 
\mathbf{Y}_{1}\mathbf{+(}1-\lambda \mathbf{)Y}_{2}\mathbf{)=}\Lambda
(\lambda (\mathbf{X}_{1},\mathbf{Y}_{1}\mathbf{)+(}1-\lambda \mathbf{)}(%
\mathbf{X}_{2},\mathbf{Y}_{2}\mathbf{)})\in \mathbb{A},
\end{equation*}%
and therefore, $\lambda \mathbf{X}_{1}\mathbf{+(}1-\lambda \mathbf{)X}_{2}\in 
\mathcal{A}^{\lambda \mathbf{Y}_{1}\mathbf{+(}1-\lambda \mathbf{)Y}_{2}}$.
\end{proof}

\

\noindent The class of systemic risk measures defined in \eqref{0011} is a fairly general representation since the aggregation function $\Lambda $ needs only
to be concave and increasing in one of its arguments and the acceptance set $%
\mathbb{A}$ is only required to be monotone and convex. As shown in the following
Corollary \ref{corAgg}, such a risk measure may describe either the
possibility of  \textquotedblleft first aggregate and second add the capital\textquotedblright (for
example if $\Lambda (\mathbf{X},\mathbf{Y):=}\Lambda _{1}(\mathbf{X)+}%
\Lambda _{2}(\mathbf{Y)}$, where $\Lambda _{2}(\mathbf{Y)}$ could be
interpreted as the discounted cost of $\mathbf{Y}$) or the case of \textquotedblleft first add and second 
aggregate\textquotedblright (for example if $\Lambda (\mathbf{X},\mathbf{%
Y):=}\Lambda _{1}(\mathbf{X+Y)}$).

\begin{corollary}
\label{corAgg}Let $\Lambda _{2}:\mathcal{C}\rightarrow \mathcal{L}^{0}(%
\mathbb{R})$ be concave, let $\mathcal{A}^{\mathbf{Y}}$ be defined in (\ref%
{AA}), where the function $\Lambda $ has one of the following forms:
\begin{align*}
\Lambda (\mathbf{Z},\mathbf{Y})&=\Lambda _{1}(\mathbf{Z})+\Lambda _{2}(
\mathbf{Y}), \\
\Lambda (\mathbf{Z},\mathbf{Y})&=\Lambda _{1}(\mathbf{Z+Y}).
\end{align*}
Then, $\mathcal{A}^{\mathbf{Y}}$ fulfills properties \ref{condmon} and \ref{condc2}. Therefore, the map $\rho $ defined in (\ref{0011}) is a
quasi-convex systemic risk measure under the assumptions \ref{condCc} and \ref%
{condpiqc}; it is a convex systemic risk measure under the assumptions \ref{condCc}
and \ref{condpic}.
\end{corollary}

\noindent We now turn to the class of truly quasi-convex systemic risk measures defined by (\ref{0012}), which represents the generalization of the quasi-convex risk
measure in (\ref{QCFA}) in the one-dimensional case. 

\begin{proposition}
Let $\theta :\mathcal{C}\rightarrow \mathbb{R}$. Then the set 
\begin{equation*}
\mathcal{A}^{\mathbf{Y}}:=\left\{ \mathbf{Z}\in \mathcal{L}^{0}(\mathbb{R}%
^{N})\mid \Lambda _{1}(\mathbf{Z)}\in \mathcal{B}^{\theta (\mathbf{Y)}%
}\right\} ,\text{\quad }\mathbf{Y\in }\mathcal{C},
\end{equation*}%
satisfies properties \ref{condmon} and \ref{condqc2}. The map $\rho $
defined in (\ref{00}) is given by 
\begin{equation}
\rho (\mathbf{X}):= \inf \{\pi (\mathbf{Y})\in \mathbb{R}\mid 
\mathbf{Y}\in \mathcal{C},\,\Lambda _{1}(\mathbf{X)}\in \mathcal{B}^{\theta (%
\mathbf{Y)}}\}\,, \label{0012}
\end{equation}%
and under the assumptions \ref{condCc} and \ref{condpiqc} is a quasi-convex systemic
risk measure.
\end{proposition}
\begin{proof}
Property \ref{condmon}: Let $\mathbf{X}_{1}\in \mathcal{A}^{\mathbf{Y}}$ and 
$\mathbf{X}_{2}\succeq \mathbf{X}_{1}$. Note that $\mathbf{X}_{1}\in 
\mathcal{A}^{\mathbf{Y}}$ implies $\Lambda _{1}(\mathbf{X}_{1}\mathbf{)}\in 
\mathcal{B}^{\theta (\mathbf{Y)}}$ and $\mathbf{X}_{2}\succeq \mathbf{X}_{1}$
implies $\Lambda _{1}(\mathbf{X}_{2}\mathbf{)\geq }\Lambda _{1}(\mathbf{X}%
_{1}\mathbf{)}$. Since $\mathcal{B}^{x}$ is a monotone set for all $x,$ we
have $\Lambda _{1}(\mathbf{X}_{2}\mathbf{)}\in \mathcal{B}^{\theta (\mathbf{%
Y)}}$\ and $\mathbf{X}_{2}\in \mathcal{A}^{\mathbf{Y}}.$ \\
Property \ref{condqc2}: Fix $\mathbf{Y}_{1},\mathbf{Y}_{2}\in \mathcal{C},$ $%
\mathbf{X}_{1}\in \mathcal{A}^{\mathbf{Y}_{1}}$, $\ \mathbf{X}_{2}\in 
\mathcal{A}^{\mathbf{Y}_{2}}\ $and $\lambda \in \lbrack 0,1]$. Then $\Lambda
_{1}(\mathbf{X}_{1}\mathbf{)}\in \mathcal{B}^{\theta (\mathbf{Y}_{1}\mathbf{)%
}}$ and $\Lambda _{1}(\mathbf{X}_{2}\mathbf{)}\in \mathcal{B}^{\theta (%
\mathbf{Y}_{2}\mathbf{)}}$. From the concavity of $\Lambda _{1}$ we obtain:%
\begin{equation*}
\Lambda _{1}(\lambda \mathbf{X}_{1}\mathbf{+(}1-\lambda \mathbf{)X}_{2})\geq
\lambda \Lambda _{1}(\mathbf{X}_{1}\mathbf{)+(}1-\lambda \mathbf{)}\Lambda
_{1}(\mathbf{X}_{2}\mathbf{)}\in \lambda \mathcal{B}^{\theta (\mathbf{Y}_{1}%
\mathbf{)}}+\mathbf{(}1-\lambda \mathbf{)}\mathcal{B}^{\theta (\mathbf{Y}_{2}.
\mathbf{)}}
\end{equation*}%
Since $(\mathcal{B}^{x})_{x\in \mathbb{R}}$ is an increasing family and each 
$\mathcal{B}^{x}$ is convex, we deduce 
\begin{equation}
\lambda \mathcal{B}^{\theta (\mathbf{Y}_{1}\mathbf{)}}+\mathbf{(}1-\lambda 
\mathbf{)}\mathcal{B}^{\theta (\mathbf{Y}_{2}\mathbf{)}}\subseteq \mathcal{B}%
^{\max \left\{ \theta (\mathbf{Y}_{1}\mathbf{),}\theta (\mathbf{Y}_{2}%
\mathbf{)}\right\} }.  \label{02}
\end{equation}%
\ Suppose that $\max (\theta (\mathbf{Y}_{1}\mathbf{),}\theta (\mathbf{Y}_{2}%
\mathbf{)})=\theta (\mathbf{Y}_{1}\mathbf{)}$, using the monotonicity of the
set $\mathcal{B}^{\theta (\mathbf{Y}_{1}\mathbf{)}}$ we deduce%
\begin{equation*}
\Lambda _{1}(\lambda \mathbf{X}_{1}\mathbf{+(}1-\lambda \mathbf{)X}_{2})\in 
\mathcal{B}^{\theta (\mathbf{Y}_{1}\mathbf{)}},
\end{equation*}%
and $\lambda \mathbf{X}_{1}\mathbf{+(}1-\lambda \mathbf{)X}_{2}\in \mathcal{A%
}^{\mathbf{Y}_{1}}$, so that property \ref{condqc2} is satisfied with $
\alpha =1.$
\end{proof}

\

\noindent In the risk measure (\ref{0012}) we are not allowing to add capital to $%
\mathbf{X}$\textbf{\ }before the aggregation takes place\textbf{, }as the
quasi-convexity property of $\rho $ would be lost in general. Next, we contemplate
this possibility (i.e. we consider conditions of the type $\,\Lambda (%
\mathbf{X},\mathbf{Y)}\in \mathcal{B}^{\theta (\mathbf{Y)}}$) in the
systemic risk measures (\ref{0013}) and (\ref{0014}) but only under some non
trivial restrictions: for the case (\ref{0013}) we impose conditions on the
aggregation function $\Lambda$ that are made explicit in equation (\ref{22})
and in the Example \ref{exampleAgg}. In contrast, for the case (\ref{0014}) we consider a
general aggregation function $\Lambda$, but we restrict the family of
acceptance sets to $\mathcal{B}^{\pi (\mathbf{Y)}},$ where $\pi$ is
positively linear and represents the risk level of the acceptance family.  

\begin{proposition}
\label{prop22}Let $\theta :\mathcal{C}\rightarrow \mathbb{R}$ \ and%
\begin{equation*}
\mathcal{A}^{\mathbf{Y}}:=\left\{ \mathbf{Z}\in \mathcal{L}^{0}(\mathbb{R}%
^{N})\mid \Lambda (\mathbf{Z},\mathbf{Y)}\in \mathcal{B}^{\theta (\mathbf{Y)}%
}\right\} ,\text{\quad }\mathbf{Y\in }\mathcal{C},
\end{equation*}%
where $\Lambda (\mathbf{\cdot },\mathbf{Y)}:\mathcal{L}^{0}(\mathbb{R}%
^{N})\rightarrow \mathcal{L}^{0}(\mathbb{R})$ is $\succeq $-increasing and
concave for all $\mathbf{Y\in }\mathcal{C}$. Assume in addition that: 
\begin{equation}
\theta (\mathbf{Y}_{2})\geq \theta (\mathbf{Y}_{1})\Rightarrow \Lambda (%
\mathbf{X},\mathbf{Y}_{2}\mathbf{)\geq }\Lambda (\mathbf{X},\mathbf{Y}_{1}%
\mathbf{)}\text{ for all }\mathbf{X\in }\mathcal{L}^{0}(\mathbb{R}^{N})\text{%
.}  \label{22}
\end{equation}%
Then properties \ref{condmon} and \ref{condqc} hold. The map $\rho $ defined
in (\ref{00}) is given by 
\begin{equation}
\rho (\mathbf{X}):= \inf \{\pi (\mathbf{Y})\in \mathbb{R}\mid 
\mathbf{Y}\in \mathcal{C},\,\Lambda (\mathbf{X},\mathbf{Y)}\in \mathcal{B}%
^{\theta (\mathbf{Y)}}\}  \label{0013}
\end{equation}%
and is a quasi-convex  systemic risk measure.
\end{proposition}
\begin{proof}
Property \ref{condmon}: Let $\mathbf{X}_{1}\in \mathcal{A}^{\mathbf{Y}}$ and 
$\mathbf{X}_{2}\succeq \mathbf{X}_{1}$. Note that $\mathbf{X}_{1}\in 
\mathcal{A}^{\mathbf{Y}}$ implies $\Lambda (\mathbf{X}_{1},\mathbf{Y)}\in 
\mathcal{B}^{\theta (\mathbf{Y)}}$ and $\mathbf{X}_{2}\succeq \mathbf{X}_{1}$
implies $\Lambda (\mathbf{X}_{2},\mathbf{Y)\geq }\Lambda (\mathbf{X}_{1},%
\mathbf{Y)}$. Since $\mathcal{B}^{\theta (\mathbf{Y)}}$ is a monotone set,
we have $\Lambda (\mathbf{X}_{2},\mathbf{Y)}\in \mathcal{B}^{\theta (\mathbf{%
Y)}}$\ and $\mathbf{X}_{2}\in \mathcal{A}^{\mathbf{Y}}.$ \\
Property \ref{condqc}: Fix $m\in \mathbb{R}$, $\mathbf{Y}_{1},\mathbf{Y}
_{2}\in \mathcal{C}$ such that $\pi (\mathbf{Y}_{1})\leq m$ and $\pi (
\mathbf{Y}_{2})\leq m$, $\lambda \in \lbrack 0,1]$ and take $\mathbf{X}
_{1}\in \mathcal{A}^{\mathbf{Y}_{1}}$ and $\mathbf{X}_{2}\in \mathcal{A}^{
\mathbf{Y}_{2}}$. Then $\Lambda (\mathbf{X}_{1},\mathbf{Y}_{1}\mathbf{)}\in 
\mathcal{B}^{\theta (\mathbf{Y}_{1}\mathbf{)}}$ and $\Lambda (\mathbf{X}_{2},
\mathbf{Y}_{2}\mathbf{)}\in \mathcal{B}^{\theta (\mathbf{Y}_{2}\mathbf{)}}$.
Then, w.l.o.g. we may assume that $\theta (\mathbf{Y}_{2})\geq \theta (\mathbf{Y}
_{1})$. Since $(\mathcal{B}^{x})_{x\in \mathbb{R}}$ is an increasing family, we have
$\mathcal{B}^{\theta (\mathbf{Y}_{1}\mathbf{)}}\subseteq \mathcal{B}^{\theta
(\mathbf{Y}_{2}\mathbf{)}}$. Condition (\ref{22}) implies $\Lambda (\mathbf{X%
}_{1},\mathbf{Y}_{2}\mathbf{)\geq }\Lambda (\mathbf{X}_{1},\mathbf{Y}_{1}%
\mathbf{)}\in \mathcal{B}^{\theta (\mathbf{Y}_{1}\mathbf{)}}\subseteq 
\mathcal{B}^{\theta (\mathbf{Y}_{2}\mathbf{)}}$, so that $\Lambda (\mathbf{X}%
_{1},\mathbf{Y}_{2}\mathbf{)}\in \mathcal{B}^{\theta (\mathbf{Y}_{2}\mathbf{)%
}}.$ From the concavity of $\Lambda (\mathbf{\cdot },\mathbf{Y}_{2}\mathbf{)}
$ and the convexity of $\mathcal{B}^{\theta (\mathbf{Y}_{2}\mathbf{)}}$ we
obtain: 
\begin{equation*}
\Lambda (\lambda \mathbf{X}_{1}\mathbf{+(}1-\lambda \mathbf{)X}_{2},\mathbf{Y%
}_{2}\mathbf{)\geq \lambda \Lambda (\mathbf{X}}_{1}\mathbf{,\mathbf{Y}}_{2}%
\mathbf{\mathbf{)+(}}1-\lambda \mathbf{\mathbf{)}\Lambda (\mathbf{X}}_{2}%
\mathbf{,\mathbf{Y}}_{2}\mathbf{\mathbf{)}}\in \mathcal{B}^{\theta (\mathbf{Y%
}_{2}\mathbf{)}}.
\end{equation*}%
Hence $\Lambda (\lambda \mathbf{X}_{1}\mathbf{+(}1-\lambda \mathbf{)X}_{2},%
\mathbf{Y}_{2}\mathbf{)}\in \mathcal{B}^{\theta (\mathbf{Y}_{2}\mathbf{)}}$
which means: $\lambda \mathbf{X}_{1}\mathbf{+(}1-\lambda \mathbf{)X}_{2}\in 
\mathcal{A}^{\mathbf{Y}_{2}}$. Since $\pi (\mathbf{Y}_{2})\leq m$, property 
\ref{condqc} holds with $\mathbf{Y}=\mathbf{Y}_{2}.$
\end{proof}

\begin{example}
\label{exampleAgg}Let $\theta :\mathcal{C}\rightarrow \mathbb{R}$ and let $%
\Lambda $ be defined by%
\begin{equation*}
\Lambda (\mathbf{X},\mathbf{Y)}=g(\mathbf{X},\theta (\mathbf{Y})),
\end{equation*}%
where $g(\cdot ,z):\mathbb{R}^{N}\rightarrow \mathbb{R}$ is  increasing and
concave for all $z\in \mathbb{R}$ and $g(x,\cdot):\mathbb{R}\rightarrow 
\mathbb{R}$ is increasing for all $x\in \mathbb{R}^N$. Then $\Lambda $ satisfies all the assumptions in
Proposition \ref{prop22}. Examples of functions $g$ satisfying these conditions are: 
\begin{equation*}
g(x,z)=f(x)+h(z),
\end{equation*}%
with $\ f$ increasing and concave and $h$ increasing, or
\begin{equation*}
g(x,z)=f(x)h(z)
\end{equation*}%
with $f$ increasing, concave and positive and $h$ increasing and positive.
\end{example}

\begin{proposition}
\label{prop:qc3} Suppose that $\mathcal{C\subseteq L}
^{0}(\mathbb{R}^{N})$ is a convex set such $\mathbf{0\in }\mathcal{C}$ and $\mathcal{C}+\mathbb{R}_{+}^{N}\in 
\mathcal{C}.$ Assume in addition that $\pi :\mathcal{C}\rightarrow \mathbb{R}
$ satisfies $\pi (u)=1$ for a given $u\in \mathbb{R}_{+}^{N}$, $u\neq \mathbf{0}$,     and
\begin{equation*}
\pi (\alpha _{1}Y_{1}+\alpha _{2}Y_{2})=\alpha _{1}\pi (Y_{1})+\alpha
_{2}\pi (Y_{2})
\end{equation*}%
for all $\alpha _{i}\in \mathbb{R}_{+}$ and $Y_{i}\in \mathcal{C}$. Let%
\begin{equation*}
\mathcal{A}^{\mathbf{Y}}:=\left\{ \mathbf{Z}\in \mathcal{L}^{0}(\mathbb{R}%
^{N})\mid \Lambda (\mathbf{Z},\mathbf{Y)}\in \mathcal{B}^{\pi (\mathbf{Y)}%
}\right\} ,
\end{equation*}%
where $\Lambda $ is concave and $\Lambda (\mathbf{X},\mathbf{\cdot }):%
\mathcal{C}\rightarrow \mathcal{L}^{0}(\mathbb{R})$ is increasing (with
respect to the componentwise ordering) for all $\mathbf{X}\in \mathcal{L}%
^{0}(\mathbb{R}^{N})$. Then the family of sets $\mathcal{A}^{\mathbf{Y}}$
fulfill properties \ref{condmon} and \ref{condqc}. The map $\rho $ defined
in (\ref{00}) is given by 
\begin{equation}
\rho (\mathbf{X})=\inf \{\pi (\mathbf{Y})\in \mathbb{R}\mid \mathbf{Y}\in 
\mathcal{C},\,\Lambda (\mathbf{X},\mathbf{Y})\in \mathcal{B}^{\pi (\mathbf{Y})}\}  \label{0014}
\end{equation}
and is a quasi-convex systemic risk measure.
\end{proposition}
\begin{proof}
Property \ref{condmon}: it follows immediately from the monotonicity of $%
\mathcal{B}^{x},$ $x\in \mathbb{R}$. \\
Property \ref{condqc}: Let $\mathbf{Y}_{1},\mathbf{Y}_{2}\in \mathcal{C}$, $%
m\in \mathbb{R}$ and assume w.l.o.g. that $\pi (\mathbf{Y}_{1})\leq \pi (%
\mathbf{Y}_{2})\leq m$. Let $\mathbf{X}_{1}\in \mathcal{A}^{\mathbf{Y}_{1}}$%
, $\ \mathbf{X}_{2}\in \mathcal{A}^{\mathbf{Y}_{2}}\ $and $\lambda \in
\lbrack 0,1]$. Then $\Lambda (\mathbf{X}_{1},\mathbf{Y}_{1}\mathbf{)}\in 
\mathcal{B}^{\pi (\mathbf{Y}_{1}\mathbf{)}}$ and $\Lambda (\mathbf{X}_{2},%
\mathbf{Y}_{2}\mathbf{)}\in \mathcal{B}^{\pi (\mathbf{Y}_{2}\mathbf{)}}$.
Because $(\mathcal{B}^{x})_{x\in \mathbb{R}}$ is increasing, we get $\Lambda (%
\mathbf{X}_{1},\mathbf{Y}_{1})\in \mathcal{B}^{\pi (\mathbf{Y}_{2}\mathbf{)}}
$. Set 
\begin{equation*}
\mathbf{\hat{Y}}_{1}:=\mathbf{Y}_{1}+(\pi (\mathbf{Y}_{2})-\pi (\mathbf{Y}%
_{1}))u\in \mathcal{C}.
\end{equation*}%
Then $\mathbf{\hat{Y}}_{1}\geq \mathbf{Y}_{1}$ and, since $\Lambda (\mathbf{X%
},\mathbf{\cdot )}$ is increasing, $\Lambda (\mathbf{X}_{1},\mathbf{\hat{Y}}%
_{1})\geq \Lambda (\mathbf{X}_{1},\mathbf{Y}_{1}\mathbf{)\in }\mathcal{B}%
^{\pi (\mathbf{Y}_{2}\mathbf{)}}$ and%
\begin{equation*}
\Lambda (\mathbf{X}_{1},\mathbf{\hat{Y}}_{1})\in \mathcal{B}^{\pi (\mathbf{Y}%
_{2}\mathbf{)}}
\end{equation*}%
because of the monotonicity of $\mathcal{B}^{\pi (\mathbf{Y}_{2}\mathbf{)}}$%
. Letting 
\begin{equation*}
\mathbf{Y}:=\lambda \mathbf{\hat{Y}}_{1}+(1-\lambda )\mathbf{Y}_{2}\in 
\mathcal{C}\,,
\end{equation*}%
and using the properties of $\pi $ we obtain: 
\begin{eqnarray*}
\pi (\mathbf{Y}) &=&\pi (\lambda \lbrack \mathbf{Y}_{1}+(\pi (\mathbf{Y}%
_{2})-\pi (\mathbf{Y}_{1}))u]+(1-\lambda )\mathbf{Y}_{2}) \\
&=&\lambda \pi (\mathbf{Y}_{1}+(\pi (\mathbf{Y}_{2})-\pi (\mathbf{Y}%
_{1}))u)+(1-\lambda )\pi (\mathbf{Y}_{2})=\pi (\mathbf{Y_{2}})\leq m.
\end{eqnarray*}%
From the concavity of $\Lambda (\mathbf{\cdot },\mathbf{\cdot )}$ and the
convexity of $\mathcal{B}^{\pi (\mathbf{Y}_{2}\mathbf{)}}$ we obtain:%
\begin{eqnarray*}
\Lambda (\lambda \mathbf{X}_{1}\mathbf{+(}1-\lambda \mathbf{)X}_{2},\mathbf{%
Y)} &=&\Lambda (\lambda \mathbf{X}_{1}\mathbf{+(}1-\lambda \mathbf{)X}%
_{2},\lambda \mathbf{\hat{Y}}_{1}+(1-\lambda )\mathbf{Y}_{2}\mathbf{)} \\
&=&\Lambda (\lambda (\mathbf{X}_{1},\mathbf{\hat{Y}}_{1}\mathbf{)+(}%
1-\lambda \mathbf{)}(\mathbf{X}_{2},\mathbf{Y}_{2}\mathbf{)}) \\
&\geq &\lambda \Lambda (\mathbf{X}_{1},\mathbf{\hat{Y}}_{1}\mathbf{)+(}%
1-\lambda \mathbf{)\Lambda }(\mathbf{X}_{2},\mathbf{Y}_{2}\mathbf{)}\\
&&\in 
\mathcal{B}^{\pi (\mathbf{Y}_{2}\mathbf{)}}=\mathcal{B}^{\pi (\mathbf{Y)}},
\end{eqnarray*}
and the monotonicity of $\mathcal{B}^{\pi (\mathbf{Y)}}$ implies: 
\begin{equation*}
\Lambda (\lambda \mathbf{X}_{1}\mathbf{+(}1-\lambda \mathbf{)X}_{2},\mathbf{%
Y)}\in \mathcal{B}^{\pi (\mathbf{Y)}},
\end{equation*}%
which means: $\lambda \mathbf{X}_{1}\mathbf{+(}1-\lambda \mathbf{)X}_{2}\in 
\mathcal{A}^{\mathbf{Y}}$. Hence, property \ref{condqc} is satisfied.
\end{proof}


\section{Scenario-dependent Allocations}\label{sec:RandCash}

We will now focus on the particularly interesting family of sets $\mathcal{C}$ of risk level vectors $\mathbf{Y}$ defined by
\begin{equation} \label{def:RandCap2}
\mathcal{C}\subseteq \{\mathbf{Y}\in \mathcal{L}^{0}(\mathbb{R}^{N})\mid
\sum_{n=1}^{N}Y^{n}\in \mathbb{R}\}=:\mathcal{C}_{\mathbb{R}}.
\end{equation}
A vector $\mathbf{Y}\in\mathcal{C}$ as in \eqref{def:RandCap2} can be interpreted as cash amount $\sum_{n=1}^{N}Y^{n}\in \mathbb{R}$ (which is known today because it is deterministic) that at the future time horizon $T$ is allocated to the financial institutions according to the realized scenario. That is, for $i=1,...,N$,  $Y^{i}(\omega)$ is allocated to institution $i$ in case scenario $\omega$ has been realized at $T$, but the total allocated cash amount $\sum_{n=1}^{N}Y^{n}$ stays constant over the different scenarios. One could think of a lender of last resort or a regulator who at time $T$ has a certain amount of cash at disposal to distribute among financial institutions in the most efficient way (with respect to systemic risk) according to the scenario that has been realized. Restrictions on the admissible distributions of cash are implied by the choice of set $\mathcal{C}$. For example, choosing $\mathcal{C}=\mathbb{R}^{N}$ corresponds to the fact that the distribution is deterministic, i.e. the allocation to each institution is already determined today, whereas for $\mathcal{C}=\mathcal{C}_{\mathbb{R}}$ the distribution can be chosen completely freely depending on the scenario $\omega$ that has been realized. Note that the latter case includes potential negative cash allocations, i.e. withdrawals of cash from certain components which allows for cross-subsidization between financial institutions. The (more realistic) situation of scenario-dependent cash distribution without cross-subsidization is represented by the set
\begin{equation*} 
\mathcal{C} := \{\mathbf{Y}\in \mathcal{C}_{\mathbb{R}} \mid
Y^{i}\ge 0, i=1,...N\}.
\end{equation*}
In this section we give some structural results and  examples concerning systemic risk measures defined in terms of sets $\mathcal{C}$ as in \eqref{def:RandCap2}. In Section \ref{sec:Gaussian} and \ref{finite_space} we then present two more extensive examples of systemic risk measures that employ specific sets  $\mathcal{C}$ of type \eqref{def:RandCap2}. \\
In the following we always assume the componentwise order relation on $\mathcal{L}^{0}(\mathbb{R}^{N})$, i.e. $\mathbf{X}_{1}\succeq \mathbf{X}_{2}$ if $\mathbf{X}_{1}^{i}\geq \mathbf{X}_{2}^{i}$ for all components $i=1,...,N$, and we start by specifying a general class of quasi-convex systemic risk measures that allow the interpretation of the minimal total amount needed to secure the system by scenario-dependent cash allocations as described above. To this end let $\mathcal{C}\subseteq \mathcal{C}_{\mathbb{R}}$ be such that 
\begin{equation}
\mathcal{C}+\mathbb{R}_{+}^{N}\in \mathcal{C}.  \label{eq:propC}
\end{equation}
Let the valuation $\pi (\mathbf{Y})$ of a $\mathbf{Y}\in\mathcal{C}$ be given by $\widetilde{\pi }(\sum_{n=1}^{N}Y^{n})$ for $\widetilde{\pi }:\mathbb{R}\rightarrow \mathbb{R}$ increasing (for example the present value of the total cash amount $\sum_{n=1}^{N}Y^{n}$ at time $T$). Further, let $(\mathcal{A}%
^{x})_{x\in \mathbb{R}}$ be an increasing family (w.r.t. $x)$ of monotone,
convex subsets $\mathcal{A}^{x}\subseteq \mathcal{L}^{0}(\mathbb{R^{N}})$,
and let $\theta :\mathbb{R}\rightarrow \mathbb{R}$ be an increasing function. We can then define the following family of systemic risk measures
\begin{equation}  \label{eq:sumsrm}
\rho (\mathbf{X}):= \inf \{\pi (\mathbf{Y})\in \mathbb{R}\mid 
\mathbf{Y}\in \mathcal{C},\,\mathbf{X}+\mathbf{Y}\in \mathcal{A}^{\theta(\sum Y^{n})}\}\,,
\end{equation}
i.e. the risk measure can be interpreted as the valuation of the minimal total amount needed at time $T$ to secure the system by distributing the cash in the most effective way among institutions. Note that here the criteria whether a system is safe or not after injecting a vector $\mathbf{Y}$ is given by the acceptance set $\mathcal{A}^{\theta(\sum Y^{n})}$ which itself depends on the total amount $\sum_{n=1}^{N}Y^{n}$. This gives, for example, the possibility of modeling an increasing level of prudence when defining safe systems for higher amounts of the required total cash. This effect will lead to truly quasi-convex systemic risk measures as the next proposition shows:

\begin{proposition}
The family of sets 
\begin{equation*}
\mathcal{A}^{\mathbf{Y}}:= \mathcal{A}^{\theta(\sum Y^{n})} - \mathbf{Y}, \ \mathbf{Y}\in\mathcal{C},
\end{equation*}
fulfills properties \ref{condmon} and \ref{condqc} with respect to the componentwise order relation on $\mathcal{L}^{0}(\mathbb{R}^{N})$. Hence the map \eqref{eq:sumsrm} is a quasi-convex risk measure. If further $\widetilde{\pi }$ is convex and $\theta$ is constant then the map \eqref{eq:sumsrm} is even a convex risk measure.
\end{proposition}
\begin{proof}
Property \ref{condmon} follows immediately from the monotonicity of $%
\mathcal{A}^{x}, x\in\mathbb{R}$.
To show Property \ref{condqc} let $\mathbf{Y}_{1},\mathbf{Y}_{2}\in \mathcal{C}$, $%
m \in \mathbb{R}$, and $\widetilde{\pi}(\sum_{n=1}^{N}Y_1^{n}) \le 
\widetilde{\pi}(\sum_{n=1}^{N}Y_2^{n}) \le m$, where w.l.o.g. $\sum Y_1^{n}
\le \sum Y_2^{n}$. Further, let $\mathbf{X}_{1}\in \mathcal{A}^{\mathbf{Y}%
_{1}}$, $\ \mathbf{X}_{2}\in \mathcal{A}^{\mathbf{Y}_{2}}\ $and $\lambda \in
\lbrack 0,1]$. Because $(\mathcal{A}^{x})_{x\in \mathbb{R}}$ and $\theta$
are increasing we get $\mathbf{X}_{1}+\mathbf{Y}_{1}\in\mathcal{A}%
^{\theta(\sum Y_2^{n})}$. Set 
\begin{equation*}
\mathbf{\hat{Y}}_{1}:= \mathbf{Y}_{1} + (\sum Y_2^{n} - \sum
Y_1^{n},0,...,0) \in \mathcal{C} \,.
\end{equation*}
Then 
\begin{equation*}
\mathbf{X}_{1}+\mathbf{\hat{Y}}_{1} \in\mathcal{A}^{\theta(\sum Y_2^{n})}
\end{equation*}
because of the monotonicity of $\mathcal{A}^{\theta(\sum Y_2^{n})}$, and 
\begin{equation*}
\lambda(\mathbf{X}_{1}+\mathbf{\hat{Y}}_{1})+(1-\lambda)(\mathbf{X}_{2}+%
\mathbf{Y}_{2}) \in\mathcal{A}^{\theta(\sum Y_2^{n})}
\end{equation*}
because of the convexity of $\mathcal{A}^{\theta(\sum Y_2^{n})}$. Furthermore,
with 
\begin{equation*}
\mathbf{Y}:=\lambda\mathbf{\hat{Y}}_{1}+(1-\lambda)\mathbf{Y}_{2} \,,
\end{equation*}
we get $\lambda\mathbf{X}_{1} + (1-\lambda)\mathbf{X}_{2} \in\mathcal{A}^{%
\mathbf{Y}}$ and $\pi(\mathbf{Y})=\pi(\mathbf{Y_2}) \le m$ since $%
\sum_{n=1}^{N}Y^{n}=\sum_{n=1}^{N}Y_2^{n}$. Hence, property \ref{condqc} is
satisfied. The final statement
follows from \cite{FrittelliScandolo}.
\end{proof}

\

\noindent Note that the quasi-convex risk measures in \eqref{eq:sumsrm} are obtained in a similar way as the ones in \eqref{prop:qc3}, the main difference being that the risk measures in \eqref{prop:qc3} are defined on an aggregated level in terms of one-dimensional acceptance sets while the ones in \eqref{eq:sumsrm} are defined in terms of general multi-dimensional acceptance sets. However, in the case $\mathcal{C}=\mathcal{C}_{\mathbb{R}}$ the next proposition shows that every systemic risk measure of type \eqref{eq:sumsrm} can be written as a univariate quasi-convex risk measure applied to the sum of the risk factors. That is, when free scenario-dependent allocations with unlimited cross-subsidization between the financial institutions are possible, the sum as aggregation rule might not only be acceptable as mentioned in the introduction but is the canonical way of aggregation and the canonical way of measuring systemic risk is of type \eqref{112}. However, while this situation and insight is relevant for a portfolio manager, the typical financial systems does not allow for unlimited cross-subsidization and more restricted sets $\mathcal{C}$ together with more appropriate aggregation rules have to be considered.

\begin{proposition}
\label{prop:sumagg} Let $\mathcal{C}=\mathcal{C}_{\mathbb{R}}$. Then $\rho $
in \eqref{eq:sumsrm} is of the form 
\begin{equation}
\rho (\mathbf{X})=\widetilde{\rho }(\sum_{n=1}^{N}X^{n})
\end{equation}%
for some quasi-convex risk measure 
\begin{equation*}
\widetilde{\rho }:\mathcal{L}^{0}(\mathbb{R})\rightarrow \overline{\mathbb{R}%
}:=\mathbb{R}\cup \left\{ -\infty \right\} \cup \left\{ \infty \right\} \,.
\end{equation*}
\end{proposition}
\begin{proof}
Let $\mathbf{X}_{1}, \mathbf{X}_{1} \in \mathcal{L}^{0}(\mathbb{R}^{N})$ be
such that $\sum_{n=1}^{N}X_1^{n}=\sum_{n=1}^{N}X_2^{n}$. In the notation of
the proof of Lemma \ref{lem:rm}, let $\mathbf{Y^1}\in B(\mathbf{X_1})$ and
set 
\begin{equation*}
\mathbf{Y_2} := \mathbf{Y_1} + (\mathbf{X_1}-\mathbf{X_2}) \in \mathcal{C}
\,.
\end{equation*}
Then $\mathbf{X_1}+\mathbf{Y_1}=\mathbf{X_2}+\mathbf{Y_2}$, and thus $
\mathbf{Y_2}\in B(\mathbf{X_2})$ because $\sum Y_1^{n}=\sum Y_2^{n}$ which
implies $\mathcal{A}^{\theta(\sum Y_1^{n})}=\mathcal{A}^{\theta(\sum
Y_2^{n})}$. Since $\pi(\mathbf{Y_1})=\pi(\mathbf{Y_2})$ this implies $\rho(
\mathbf{X_1})\ge\rho(\mathbf{X_2})$. Interchanging the roles of $\mathbf{X_1}
$ and $\mathbf{X_2}$ yields $\rho(\mathbf{X_1}) = \rho(\mathbf{X_2})$, and
the map $\widetilde{\rho} :\mathcal{L}^{0}(\mathbb{R})\rightarrow \overline{
\mathbb{R}}$ given by 
\begin{equation*}
\widetilde{\rho}(X) := \rho(\mathbf{X})\,,
\end{equation*}
where $\mathbf{X} \in \mathcal{L}^{0}(\mathbb{R^N})$ is such that $X=\sum_{n=1}^{N}X^{n}$ is well-defined. For $X_1,X_2 \in \mathcal{L}^{0}(
\mathbb{R})$ define 
\begin{equation*}
\mathbf{X_i}:=(X_i,0,...,0) \in \mathcal{L}^{0}(\mathbb{R^N}) \;, \ \ i=1,2
\,.
\end{equation*}
Then 
\begin{align*}
\widetilde{\rho}(\lambda X_1 + (1-\lambda) X_2) = \rho(\lambda \mathbf{X_1}
+ (1-\lambda) \mathbf{X_2})& \le \max\{\rho(\mathbf{X_1}),\rho(\mathbf{X_2})\} \\
&= \max\{\widetilde{\rho}(X_1),\widetilde{\rho}(X_2)\}\,.
\end{align*}
Further, if $X_1\le X_2$ then $\mathbf{X_1}\le\mathbf{X_2}$ and 
\begin{equation*}
\widetilde{\rho}(X_1)=\rho(\mathbf{X_1})\ge \rho(\mathbf{X_2}) = \widetilde{%
\rho}(X_2) \,.
\end{equation*}
So $\widetilde{\rho} :\mathcal{L}^{0}(\mathbb{R})\rightarrow \overline{%
\mathbb{R}}$ is a quasi-convex risk measure and $\rho(\mathbf{X}) = 
\widetilde{\rho} (\sum_{n=1}^{N}X^{n})$.
\end{proof}

\ 

\noindent We conclude this Section by two examples that compare the risk measurement by \textquotedblleft injecting after aggregation\textquotedblright as in \eqref{2} versus the risk measurement by  \textquotedblleft injecting before aggregation\textquotedblright as in \eqref{4} for different sets $\mathcal{C}\subset\mathcal{C}_{\mathbb{R}}$ in the situation of the worst case and the expected shortfall acceptance sets, respectively.

\subsection{Example: Worst Case Acceptance Set
\label{C-worst}}
In this example we measure systemic risk by considering aggregated risk factors defined in terms of the aggregation rule
\[
\Lambda(\mathbf{X}):=\sum_{i=1}^{N}-(X_{i})^{-}.
\]
Further, we consider the acceptance set $\mathbb{A}^W$ associated to the worst case risk measure, that is a system $\mathbf{X}$ is acceptable (or safe) if $\sum_{i=1}^{N}-(X_{i})^{-} \in \mathbb{A}^W$ where $\mathbb{A}^W:=\mathcal{L}_+^{0}(\mathbb{R})$, and we denote by $\rho_W:\mathcal{L}^{0}(\mathbb{R})\rightarrow \overline{\mathbb{R}}$ the univariate worst case risk measure defined by
\[
\rho_W(X):=\inf \left\{m \in \mathbb{R} \ |\ X+ m \in \mathbb{A}^W\right\} \,.
\]
The possible sets $ \mathcal{C}$ are on one hand the deterministic allocations $ \mathcal{C}=\mathbb{R}^{N}$ and on the other hand the family of constrained scenario-dependent cash allocations of the form
\begin{align*}
& \mathcal{C}_{\gamma }:=\left\{ \mathbf{Y}\in \mathcal{C}_{\mathbb{R}}\ |\ Y_{i}\geq
\gamma _{i}\,, i=1,...N\right\} \,,
\end{align*}%
where $\gamma :=(\gamma _{1},...,\gamma _{N}),\,\gamma _{i}\in \lbrack
-\infty ,0]$. Note that for $\gamma :=(-\infty ,...,-\infty )$ this family of subsets
includes $\mathcal{C}_{\infty }=\mathcal{C}_{\mathbb{R}}$. Finally, we let the valuation be 
\[
\pi (\mathbf{Y}):=\sum_{i=1}^{n}Y_{i}.
\]
The objective of the following proposition is to analyze and relate the systemic risk measurement by \textquotedblleft injecting cash after aggregation\textquotedblright:
\begin{equation*}
\rho^{ag}(\mathbf{X}):=\inf \left\{y\in\mathbb{R} \ |\  \Lambda(\mathbf{X})+ y \in \mathbb{A}^W
\right\}=\rho_W(\sum_{i=1}^N -(X^i)^-) \,,
\end{equation*}
to the systemic risk measurement by \textquotedblleft injecting cash before aggregation\textquotedblright, both in the case of deterministic cash allocations:
\begin{equation*}
\rho^{\mathbb{R}^N}(\mathbf{X}):=\inf \left\{\pi (\mathbf{Y})|\ \mathbf{Y}\in {\mathbb{R}^N} \,, \Lambda(\mathbf{X}+\mathbf{Y}) \in \mathbb{A}^W \right\} \,,
\end{equation*}
as well as in the case of scenario-dependent cash allocations:
\begin{equation*}
\rho^{\gamma}(\mathbf{X}):=\inf \left\{\pi (\mathbf{Y}) |\ \mathbf{Y}\in \mathcal{C}%
_{\gamma}\,, \Lambda(\mathbf{X}+\mathbf{Y}) \in \mathbb{A}^W\right\} .
\end{equation*}

\begin{proposition}
\label{pro:ex} It holds that 
\begin{align*}
\rho^{\mathbb{R}^N}(\mathbf{X})& =\sum_{i=1}^N \rho_W(X^i) \geq \rho^{ag}(\mathbf{X}) \\
\rho^{\gamma}(\mathbf{X})& =\rho_W\left(\sum_{i=1}^N (X^i \mathbb{I
}_{\left\{X^i\le -\gamma_i\right\}} - \gamma_i \mathbb{I}_{\left\{X^i\ge -\gamma_i\right\}}) \right)\le \rho^{ag}(\mathbf{X})\,.
\end{align*}
In particular, for $\gamma=0:=(0,...,0)$ we get $\rho^{0}(\mathbf{X})=\rho^{ag}(\mathbf{X})$, and for $\gamma=-\infty:=(-\infty ,...,-\infty )$ we get $\rho^{-\infty}(\mathbf{X})=%
\rho_W(\sum_{i=1}^N X^i )$.
\end{proposition}
Before we prove the proposition we make some comments on the results. We see that if we interpret the risk measure as capital requirement (which in this situation also is possible for $\rho^{ag}$ since the aggregation $ \Lambda(\mathbf{X})$ can be interpreted as a monetary amount), the capital requirement when \textquotedblleft injecting before aggregation\textquotedblright with deterministic allocations is higher than the one when \textquotedblleft injecting after aggregation\textquotedblright . When allowing for \textquotedblleft injecting before aggregation\textquotedblright  with scenario-dependent cash allocations, the gained flexibility in allocating the cash leads to decreasing capital requirements. For fully flexible allocations the minimum amount $\rho^{-\infty}(\mathbf{X})=\rho_W(\sum_{i=1}^N X^i )$ is obtained, which corresponds to the representation given in Proposition \ref{prop:sumagg} in terms of the sum as aggregation rule. Obviously, here the relations between $\rho^{ag}$, $\rho^{\mathbb{R}^N}$, and $\rho^{\gamma}$ depend on the choice of the acceptance set in conjunction with the aggregation function as is illustrated in the next example.

Further, from the proof below it follows that in the case $\mathcal{C}=\mathbb{R}^{N}$ there exists a unique allocation $\mathbf{Y}^*\in\mathbb{R}^{N}$ for a given $\mathbf{X}\in\mathcal{L}^{0}(\mathbb{R}^{N})$ such that $\rho^{\mathbb{R}^N}(\mathbf{X})=\pi (\mathbf{Y}^*)$, which implies an unambiguous ranking of the systemic riskiness of the institutions. On the other hand, in the case $\mathcal{C}=\mathcal{C}_\gamma$ there generically exist infinitely many scenario-dependent allocations $\mathbf{Y}^*\in\mathcal{C}_\gamma$ for a given $\mathbf{X}\in\mathcal{L}^{0}(\mathbb{R}^{N})$ for which the infimum of the risk measure $\rho^{\gamma}(\mathbf{X})=\pi (\mathbf{Y}^*)$ is obtained. In that case one needs to discuss further how to pick an allocation and to establish a ranking of systemic riskiness of the institutions. \\

\noindent \begin{proof}
Note that for $\mathbf{X}\in\mathcal{L}^{0}(\mathbb{R}^{N})$ it holds that $%
\Lambda(\mathbf{X})\in\mathbb{A}^W$ iff $X^i\in \mathbb{A}^W,\, i=1,...,N$. Thus
we can rewrite 
\begin{align*}
\rho^{\mathbb{R}^N}(\mathbf{X})& :=\inf \left\{\sum_{i=1}^N Y^i | \mathbf{Y}\in 
\mathbb{R}^N\,, \mathbf{X}+\mathbf{Y} \in (\mathbb{A}^W)^N \right\} \,,
\end{align*}
and obviously get 
\begin{align*}
\rho^{\mathbb{R}^N}(\mathbf{X})& = \sum_{i=1}^N -\text{ess.inf} (X^i) = \sum_{i=1}^N
\rho_W(X^i),
\end{align*}
and for $\mathbf{X}\in\mathcal{L}^{0}(\mathbb{R}^{N})$ the allocation $\mathbf{\hat{Y}}:=(\text{ess.inf} (X^1),...,\text{ess.inf} (X^N))$ is the unique $\mathbf{\hat{Y}}\in\mathbb{R}^{N}$ such that $\rho^{\mathbb{R}^N}(\mathbf{X})=\pi (\mathbf{\hat{Y}})$. \\
For $\rho^{\gamma}$ we analogously rewrite
\begin{align*}
\rho^{\gamma}(\mathbf{X})& :=\inf \left\{\sum_{i=1}^N Y^i | \mathbf{Y}\in 
\mathcal{C}_{\gamma}\,, \mathbf{X}+\mathbf{Y} \in (\mathbb{A}^W)^N \right\} .
\end{align*}
Now consider first the optimization problem 
\begin{align}  \label{eq:auxrho}
\widetilde{\rho}(\mathbf{X}):=\inf \left\{ \text{ess.sup} (\sum_{i=1}^N Y^i) | \mathbf{Y} \in 
\mathcal{L}^{0}(\mathbb{R}^{N}), Y^i \geq \gamma_i \,, \mathbf{X}+\mathbf{Y} \in (\mathbb{A}^W%
)^N \right\} .
\end{align}
Then clearly $\widetilde{\rho} \le \rho^{\gamma}$ and $\mathbf{Y}^* := -(X^i \mathbb{I
}_{\left\{X^i\le -\gamma_i\right\}} - \gamma_i \mathbb{I}_{\left\{X^i\ge -\gamma_i\right\}} )_{i=1,...N} $
is an optimal solution of \eqref{eq:auxrho}. Now define 
\begin{equation*}
\widetilde{\mathbf{Y}} := \mathbf{Y}^* + (\text{ess.sup} (\sum_{i=1}^N Y^*_i)-\sum_{i=1}^N
Y^*_i,0,...,0).
\end{equation*}
Then $\widetilde{\mathbf{Y}} \in \mathcal{C}_{\gamma}$ and $\rho^{\gamma}(\mathbf{X}) \le \pi(%
\widetilde{\mathbf{Y}}) =\text{ess.sup} (\sum_{i=1}^N Y^*_i)=\widetilde{\rho}(\mathbf{X})\le
\rho^{\gamma}(\mathbf{X}) $, and thus 
\begin{align*}
\rho^{\gamma}(\mathbf{X}) &= \sum_{i=1}^N \widetilde{Y}_i = \text{ess.sup}
\left(\sum_{i=1}^N Y^*_i\right) \\
&= \text{ess.sup} \left(\sum_{i=1}^N -(X^i \mathbb{I
}_{\left\{X^i\le -\gamma_i\right\}} - \gamma_i \mathbb{I}_{\left\{X^i\ge -\gamma_i\right\}})\right) \,.
\end{align*}
Finally we remark that generically for a given $\mathbf{X}\in\mathcal{L}^{0}(\mathbb{R}^{N})$ the above allocation $\widetilde{\mathbf{Y}}\in \mathcal{C}_{\gamma}$ is not unique such that $\rho^{\gamma}(\mathbf{X})=\pi (\widetilde{\mathbf{Y}})$. In fact, any allocation of the form
\begin{equation*}
\mathbf{Y}^* + (Z_1,...,Z_N)
\end{equation*}
with $(Z_1,...,Z_N)\in\mathcal{L}^{0}(\mathbb{R}^{N})$ such that $\sum_{i=1}^N Z_i=\text{ess.sup} (\sum_{i=1}^N Y^*_i)-\sum_{i=1}^N Y^*_i$ will satisfy the desired property.
\end{proof}

\subsection{Example:  Expected Shortfall Acceptance Set
\label{ES}}

Now consider the acceptance set associated to the \textquotedblleft  Expected Shortfall\textquotedblright
risk measure $\rho_{ES}$ (at some given quantile level): 
\begin{equation*}
\mathbb{A}^{ES} :=\{X\in\mathcal{L}^{0}(\mathbb{R})\ |\  \rho_{ES}(X) \leq 0\}.
\end{equation*}
See e.g. \cite{FollmerSchied2} for the definition of $\rho_{ES}$. Everything else is assumed to be as in Example \ref{C-worst}. Then 
\begin{align*}
\rho^{ag}(\mathbf{X}) = \rho_{ES}(\sum_{i=1}^N -(X^i)^-) \,.
\end{align*}
For $\rho^{\mathbb{R}^N}$ and $\rho^{\gamma}$, however, $\mathbb{A}^{ES}$ gives the same result as $\mathbb{A}^{W}$, i.e. 
\begin{align}
\rho^{\mathbb{R}^N}(\mathbf{X})& =\sum_{i=1}^N \rho_W(X^i) \geq \rho^{ag}(\mathbf{X}) \label{eq:ES1}\\
\rho^{\gamma}(\mathbf{X}) &=\rho_W\left(\sum_{i=1}^N (X^i \mathbb{I
}_{\left\{X^i\le -\gamma_i\right\}} - \gamma_i \mathbb{I}_{\left\{X^i\ge -\gamma_i\right\}})\right)\,. \label{eq:ES2}
\end{align}
Indeed, by the definition of $\rho_{ES}$ it immediatly follows that $\sum_{i=1}^N -(X^i)^- \in 
\mathbb{A}^{ES}$ if and only if $X^i\in \mathbb{A}^{W},\, i=1,...,N$, and \eqref{eq:ES1} and \eqref{eq:ES2} is then obtained from Proposition \ref{pro:ex}. So opposite to the situation in Example \ref{C-worst}, here the risk measure  when \textquotedblleft injecting before aggregation\textquotedblright  even with scenario-dependent  allocations might be higher than the one  when \textquotedblleft injecting after aggregation\textquotedblright . Indeed, we easily see that we always have $\rho^{0}\ge\rho^{ag}$, and generically even $\rho^{-\infty}\ge\rho^{ag}$ holds. This illustrates that these kind of relations highly depends on the interplay between aggregation and acceptance set.



\section{Gaussian Systems \label{sec:Gaussian}}

In this Section we assume a Gaussian financial system, i.e.~we let $\mathbf{X}=(X^1,\cdots,X^N)
$ be an $N$-dimensional Gaussian random vector with covariance matrix $Q$,
where $[Q]_{ii}:= \sigma_i^2$, $i=1,\cdots,N$, and $[Q]_{ij}:= \rho_{i,j}$
for $i\neq j$, $i,j=1,\cdots,N$, and mean vector $\mu:=(\mu_1,\cdots,\mu_N)$%
, i.e. $\mathbf{X}\sim N(\mu,Q)$. The systemic risk measure we consider is given by
\begin{equation}\label{GRM}
\rho(\mathbf{X}):=\inf\left\{\sum_{i=1}^N Y^i \ |\ \mathbf{Y}\in\mathcal{C} \subseteq \mathcal{C}_{\mathbb{R}}\,, \ \Lambda(\mathbf{X}+\mathbf{Y}) \in \mathbb{A}_{\gamma} \right\},
\end{equation}
where the set $\mathcal{C}_{\mathbb{R}}$ of scenario-dependent cash allocations is defined in \eqref{def:RandCap2}, the aggregation rule is given by $\Lambda(\mathbf{X}):=\sum_{i=1}^N -(X^i -d_i)^-$ for $d_i\in \mathbb{R}$, and the acceptance set is
\begin{equation}\label{set_gamma}
\mathbb{A}_{\gamma}:=\left\{Z \in \mathcal{L}^0(\mathbb{R})\ | \ \esp{Z}%
\geq -\gamma\right\}
\end{equation}
for some $\gamma \in \mathbb{R}^+$. Here, $d_i$ in the aggregation rule denotes some critical liquidity level of institution $i$, $i=1,...N$, and the risk measure is concerned with the expected total shortfall below these levels in the system. In Subsection~\ref{sec:Gaussian constant injection} we compute the allocation and the systemic risk measure in case of deterministic cash allocations $\mathcal{C}:=\mathbb{R^N}$, and in Subsection~\ref{sec:Gaussian random injection} we allow for more flexible scenario-dependent allocations of the form
\begin{equation}\label{GRandAll}
\mathcal{C}:=\left\{ \mathbf{Y}\in \mathcal{L}^0(\mathbb{R}^n)\ |\  \mathbf{Y}= \mathbf{m} +\mathbf{\alpha} I_D,\, \mathbf{m} ,\mathbf{\alpha}  \in \mathbb{R}^N, \sum_{i=1}^N \alpha_i =0\right\} \subseteq\mathcal{C}_{\mathbb{R}},
\end{equation}
where  $I_D$ is the indicator function of the event $D:=\left\{ \sum_{i=1}^N X^i \leq d\right\}$ for some $d\in \mathbb{R}$. Note that the condition $\sum_{i=1}^N \alpha_i =0$ implies that $\sum_{i=1}^N Y^i$ is constant a.s. Cash allocations in \eqref{GRandAll} can be interpreted as the flexibility to let the allocation depend on whether the system at time $T$ is in trouble or not, represented by the events that $\sum_{i=1}^N X^i$ is less or greater than some critical level $d$, respectively. In Subsection \ref{sec:deterministic injection} we then apply the results to a Gaussian system that is interconnected by the flow of capital between the institutions through a system of interacting Ornstein-Uhlenbeck diffusions.

\subsection{Deterministic Cash Allocations \label{sec:Gaussian constant injection}} 

We now consider the case $\mathcal{C}=\mathbb{R^N}$ and we are interested in computing the systemic risk measure
\begin{equation}\label{eq:DetRM}
\rho(\mathbf{X}):=\inf\left\{\sum_{i=1}^N m_i \ |\ \, \mathbf{m}=(m_1,\cdots,m_N)\in\mathbb{R^N} \,, \ \Lambda(\mathbf{X}+\mathbf{m}) \in \mathbb{A}_{\gamma}\right\},
\end{equation}
where for notational clarity we write $\mathbf{m}$ instead of $\mathbf{Y}$ for deterministic cash allocations. We thus need to minimize the objective function $\sum_{i=1}^N m_i$ over $\mathbb{R^N}$ under the constrained $\Lambda(\mathbf{X}+\mathbf{m}) \in \mathbb{A}_{\gamma}$, which clearly is equivalent to the constraint
\begin{equation}\label{eq:Constr}
\sum_{i=1}^N \esp{(X^i +m_i-d_i)^-} = \gamma \,.
\end{equation}
This constrained optimization problem can be solved with the associated Lagrangian
\begin{equation}\label{eq:Lagr}
L(m_1,...,m_N,\lambda):= \sum_{i=1}^N m_i + \lambda(\sum_{i=1}^N\psi_i(m_i) -\gamma)
\end{equation}
where $\psi_i(m_i):=\esp{(X^i +m_i-d_i)^-}$.
Since $X^i\sim N(\mu_i,\sigma_i^2 )$, one obtains for $i=1,...,N$ that
\begin{equation}\label{eq_neg_part}
\psi_i(m_i)= \frac{\sigma_i}{\sqrt{2\pi}}\exp{\left[-\frac{(d_i-\mu_i-m_i)^2}{
2\sigma_i^2}\right]} - (m_i+\mu_i-d_i) \Phi(\frac{d_i-\mu_i-m_i}{\sigma_i}),
\end{equation}
where $\Phi(x)=\int_{+\infty}^x \frac{1}{\sqrt{2\pi}}e^{-t^2/2} dt$.
By direct computation this leads to 
\begin{equation}\label{eq_phi}
\frac{\partial L(m_1,...,m_N,\lambda)}{\partial m_i}= 1+ \lambda \Phi(\frac{d_i-\mu_i-m_i}{\sigma_i}).
\end{equation}
By solving the Lagrangian system we then obtain the critical
point $\mathbf{m}^*=(m^*_1,\cdots,m^*_N)$ given by
\begin{equation*}
m_i^*=d_i-\mu_i -\sigma_i R,
\end{equation*}
where $R$ solves the equation
\begin{equation}  \label{optR}
P(R):=R\Phi(R)+\frac{1}{\sqrt{2\pi}}\exp{\left[-\frac{R^2}{2}\right]} =
\frac{\gamma}{\sum_{i=1}^N\sigma_i}.
\end{equation}
It is easily verified that $\mathbf{m}^*$ is indeed a global minimum and thus the optimal cash allocation associated with the risk measure \eqref{eq:DetRM}. The unique optimal cash allocation $\mathbf{m}^*$ now also induces a ranking of the institutions according to systemic riskiness, and we can discuss the dependence of this ranking with respect
to $\mu_i$ and $\sigma_i$:

\begin{enumerate}
\item $\frac{\partial m_i}{\partial \mu_i}=-1$: the systemic riskiness decreases with increasing mean.

\item $\frac{\partial m_i}{\partial \sigma_i}>0$: the systemic riskiness increases with increasing volatility. In order to show $\frac{\partial m_i}{\partial \sigma_i}>0$ we first note that $R$ is a solution of \eqref{optR} if
and only if $R$ is negative. Indeed, for $R \geq 0$ the
left-hand side of \eqref{optR} is always strictly positive,  the
right-hand side is negative. Thus
\begin{equation}  \label{derSigma}
\frac{\partial m_i}{\partial \sigma_i}= -R -\sigma_i\frac{\partial R}{
\partial \sigma_i}.
\end{equation}
By differentiating \eqref{optR} we obtain
\begin{equation}
\frac{\partial P}{\partial \sigma_i}=\frac{\partial P}{\partial R}\frac{\partial R}{\partial \sigma_i}= -\frac{
\gamma}{(\sum_{k=1}^N \sigma_k)^2}.
\end{equation}
Since $\frac{\partial P}{\partial R}=\Phi(R)$, we can compute $\frac{
\partial R}{\partial \sigma_i}$ and substitute it in \eqref{derSigma}:
\begin{align*}
\frac{\partial m_i}{\partial \sigma_i}&= -R +\frac{\sigma_i \gamma}{
(\sum_{k=1}^N \sigma_k)^2} \frac{1}{\Phi(R)} \\
&=-R +\frac{\sigma_i (\sum_{k=1}^N \sigma_k)P(R)}{(\sum_{k=1}^N
\sigma_k)^2\Phi(R)} \\
&= -R +\frac{\sigma_i (R\Phi(R)+\frac{1}{\sqrt{2\pi}}\exp{\left[-\frac{R^2}{2%
}\right]})}{(\sum_{k=1}^N \sigma_k)\Phi(R)} \\
&= (\frac{\sigma_i}{\sum_{k=1}^N \sigma_k}-1) R + \frac{\sigma_i}{%
\sum_{k=1}^N \sigma_k} \frac{1}{\sqrt{2\pi} \Phi(R)} \exp{\left[-\frac{R^2}{2%
}\right]}.
\end{align*}
Since $R$ must be negative, this implies $\frac{\partial m_i}{\partial \sigma_i}>0$.
\end{enumerate}

\subsection{A Class of Scenario-Dependent Allocations \label{sec:Gaussian random injection}}
We now allow for different allocations of the total capital at disposal depending on which state the system is in. More precisely, we differentiate between the two states that $D:= \{S\le d\}$ and $D^c=\{S\ > d\}$ for some level $d\in\mathbb{R}$ and $S:=\sum_{i=1}^N X^i$, and consider allocations $\mathcal{C}$ given in \eqref{GRandAll}. The systemic risk measure now becomes
\begin{equation*}
\rho(\mathbf{X}):=\inf\left\{\sum_{i=1}^N m_i | \ \mathbf{m} +\mathbf{\alpha} I_D \in \mathcal{C}\,, \,
\Lambda(\mathbf{X}+\mathbf{m} +\mathbf{\alpha} I_D) \in \mathbb{A}_{\gamma}\right\}.
\end{equation*}
To compute the risk measure in this case we now need to minimize the objective function $\sum_{i=1}^N m_i$ over $(\mathbf{m},\mathbf{\alpha})\in\mathbb{R}^{2 N}$ under the constraints
\begin{equation*}
\sum_{i=1}^N \alpha_i = 0   \text{ \ \ \ and \ \ \ }   \sum_{i=1}^N \esp{(X^i +m_i + \alpha_i I_D - d_i)^-} = \gamma \,.
\end{equation*}
In analogy to the Section~\ref{sec:Gaussian random injection} we apply the method of Lagrange multipliers to minimize the function

\begin{eqnarray}\label{Lagrange_gauss}
&&\phi(m_1,\cdots,m_N, \alpha_1,\cdots, \alpha_{N-1},\lambda)=\nonumber
\\&&\quad
\sum_{i=1}^N m_i + \lambda \left( \Psi(m_1,\cdots,m_N, \alpha_1,\cdots, \alpha_{N-1})-\gamma\right),
\end{eqnarray}
where 
\begin{align*}
&\Psi(m_1,\cdots,m_N, \alpha_1,\cdots, \alpha_{N-1}):= \\
& \quad \sum_{i=1}^{N-1} \esp{(X^i +m_i + \alpha_i I_D - d_i)^-}  +  
\esp{(X^N+m_N - \sum_{j=1}^{N-1}\alpha_j I_D - d_N)^-},
\end{align*}
as follows.
\begin{enumerate}
\item \emph{By computing the derivatives with respect to $\alpha_i$, $i=1,\cdots,N-1$}: 
$\frac{\partial \phi}{\partial \alpha_i}=0$ if and only if

\begin{equation} \label{eq_F}
F_{i,S}(d_i -m_i-\alpha_i, d) =F_{N,S}(d_N -m_N+\sum_{j=1}^{N-1}\alpha_j, d)
\end{equation}
for $i=1,\cdots,N-1$, where $F_{i,S}$ and $F_{N,S}$ are the joint distribution functions  of $(X^i,S)$ and $(X^N,S)$ respectively.

\item  \emph{By computing the derivatives with respect to $m_i$, for $i=1,\cdots,N-1$}: 
$\frac{\partial \phi}{\partial m_i}=0$ if and only if 
\begin{eqnarray}\label{eq_m}
&&\Phi(\frac{d_i-\mu_i-m_i}{\sigma_i}) +  F_{i,S}(d_i -m_i, d)=\nonumber
\\
&&  \Phi(\frac{d_N-\mu_n-m_N}{\sigma_n}) +  F_{N,S}(d_N-m_N, d),
\end{eqnarray}
for $i=1,\cdots,N-1$.
\item \emph{By computing the derivatives with respect to $\lambda$}: $\frac{\partial \phi}{\partial \lambda}=0$ if and only if 
$\Psi(m_1,\cdots,m_N, \alpha_1,\cdots, \alpha_{N-1})=\gamma$, where 

\begin{align*}
&\Psi(m_1,\cdots,m_N, \alpha_1,\cdots, \alpha_{N-1}) =\sum_{i=1}^N \psi_i(m_i) \\
&   + \sum_{i=1}^{N-1} \left[(m_i -d_i) F_{N,S}(d_i -m_i, d)  -(m_i+\alpha_i
-d_i)F_{i,S}(d_i -m_i-\alpha_i, d) \right.\\
& \quad \left.+\int_{d_i -m_i-\alpha_i}^{d_i -m_i}\int_{-\infty}^d
xF_{i,S}(x,y) dydx\right] + (m_N -d_N) F_{N,S}(d_N -m_N, d)\\
& \quad   -(m_N-\sum_{j=1}^{N-1}\alpha_j
-d_N)F_{N,S}(d_N -m_N+\sum_{j=1}^{N-1}\alpha_j, d) \\
& \quad \quad +\int_{d_N-m_N+\sum_{j=1}^{N-1}\alpha_j}^{d_N -m_N}\int_{-\infty}^d
xF_{N,S}(x,y) dydx,
\end{align*}
where $\psi_i$, $i=1,\cdots,N$, are defined in \eqref{eq_neg_part}.

\end{enumerate}

\noindent From \eqref{Lagrange_gauss} and \eqref{eq_F} we immediately obtain that if the $X^i$, $i=1,\cdots,N$, are identically distributed, then the optimal solution is obtained for $\alpha_i=0$, $i=1,\cdots,N$, and corresponds to the one obtained explicitly in Section \ref{sec:Gaussian constant injection} for deterministic injections. \newline
We now present numerical illustrations of our results in the simple case with two banks. \newline
In Table \ref{table1} we set
the means $\mu_i=0$ for $i=1,2$, the standard deviations $\sigma_1=1$, $\sigma_2=3$,  the acceptance level $\gamma=0.7$ and the critical level $d=2$. The last 2 columns show the sensitivities with respect to the correlation  for   deterministic allocation (case $\alpha=0$, computed in Section \ref{sec:Gaussian constant injection}) and for scenario-dependent allocation, respectively.
We observe that for highly positively correlated banks the scenario-dependent allocation does not change the total capital requirement $m_1+m_2$. Indeed, as expected, if the banks are moving together, one may have to subsidize both of them.  However, when they are negatively correlated, one benefits from
 scenario-dependent allocation since the total allocation $m_1+m_2$ is lower in that case.

\begin{table}
\centering
\begin{tabular}{|c|c|c|c|}
\hline
$\rho_{1,2}\downarrow$&&Deterministic&Random\\
\hline
&$m_1$&0.5772&0.1597\\
&$m_2$&1.7316&1.7230\\
-0.8&$\alpha$&0&2.8704\\
&$\rho=m_1+m_2$&2.3088&1.8827\\
\hline
&$m_1$&0.5772&0.2908\\
&$m_2$&1.7316&1.7776\\
-0.5&$\alpha$&0&2.3161\\
&$\rho=m_1+m_2$&2.3088&2.0683\\
\hline
&$m_1$&0.5772&0.4490\\
&$m_2$&1.7316&1.7796\\
0&$\alpha$&0&1.7208\\
&$\rho=m_1+m_2$&2.3088&2.2286\\
\hline
&$m_1$&0.5772&0.5463\\
&$m_2$&1.7316&1.7461\\
0.5&$\alpha$&0&1.3389\\
&$\rho=m_1+m_2$&2.3088&2.2924\\
\hline
&$m_1$&0.5772&0.5737\\
&$m_2$&1.7316&1.7314\\
0.8&$\alpha$&0&0.7905\\
&$\rho=m_1+m_2$&2.3088&2.3053\\
\hline
\end{tabular}
\caption{Sensitivity with respect to correlation.}
\label{table1}
\end{table}

In Table \ref{table2} we set
the means $\mu_i=0$ for $i=1,2$, the correlation $\rho=-0.5$, the standard deviation $\sigma_1=1$,  the acceptance level $\gamma=0.7$ and the critical level $d=2$ and we show sensitivity with respect to the standard deviation $\sigma_2$ of the second bank.
We observe that for equal marginals ($\sigma_1=\sigma_2=1$) random  allocation does not change the total capital requirement, as already stated in Section \ref{sec:Gaussian random injection}.  As $\sigma_2$  increases,  the systemic  risk measure increases and the allocation increases with increasing standard deviation in agreement with the sensitivity analysis  presented in Section \ref{sec:Gaussian constant injection} for the deterministic case.
Also we observe that scenario-dependent  allocation allows for smaller
total capital requirement  $m_1+m_2$.

\begin{table}
\centering
\begin{tabular}{|c|c|c|c|}
\hline
$\sigma_2\downarrow$&&Deterministic&Random\\
\hline
&$m_1$&0.1008&0.1008\\
&$m_2$&0.1031&0.1031\\
1&$\alpha$&0&0.0002\\
&$\rho=m_1+m_2$&0.2039&0.2039\\
\hline
&$m_1$&0.8168&0.3167\\
&$m_2$&4.0816&4.1295\\
5&$\alpha$&0&3.5987\\
&$\rho=m_1+m_2$&4.8984&4.4462\\
\hline
&$m_1$&1.1417&0.4631\\
&$m_2$&11.3964&11.4333\\
10&$\alpha$&0&6.9909\\
&$\rho=m_1+m_2$&12.5381&11.8963\\
\hline
\end{tabular}
\caption{Sensitivity with respect to standard deviation. }
\label{table2}
\end{table}

\subsection{Application to Models of  Borrowing and Lending}\label{sec:deterministic injection}

We consider the Gaussian vector $X_t=(X^i_t, i=1,\cdots, n)$ generated by the following dynamics
\begin{align}\label{correlated coupled}
&dX_{t}^{i}=\left[\sum_{j=1}^{N}p_{i,j}(X_{t}^{j}-X_{t}^{i})\right]dt+\sigma^i\left(\rho_i
dW^0_{t}+\sqrt{1-\rho_i^{2}}dW_{t}^{i}\right),\, i=1,\cdots,N,
\end{align}
where $\left(W^0_t,W_t^i,i=1,\cdots, N\right)$ are independent standard Brownian motions and $W^0_{t}$ is a common noise.
The lending-borrowing preferences $p_{i,j}$ are nonnegative and symmetric: $p_{i,j}=p_{j,i}$. This model is studied in detail
in \cite{CarmonaFouqueSun}
in the mean-field context where $p_{i,j}=p/N$ with $p\geq 0$. It is shown that in that case the dynamics (\ref{correlated coupled}) emerges as a Nash equilibrium of a specific stochastic game.
%

%




The total \textquotedblleft  capitalization\textquotedblright is given by:
\begin{align*}
&\sum_{i=1}^{N}X_{t}^{i}= \sum_{i=1}^{N}x_{0}^{i}+\left(\sum_{i=1}^N\sigma^i\rho_i\right)
W^0_{t}+\sum_{i=1}^{N}\sigma^i \sqrt{1-\rho_i^{2}}W_{t}^{i}\\
&\overset{{\cal D}}{=} \sum_{i=1}^{N}x_{0}^{i}+\alpha B_{t},
\end{align*}
where $B_{t}$ is a standard Brownian motion and
\[
\alpha^2=\left(\sum_{i=1}^N \sigma^i\rho_i\right)^2+\sum_{i=1}^N (\sigma^i)^2(1-\rho_i^2).
\]
We are interested in the system at a given time $t>0$ and in quantities such as the liquidity available at time $t$ defined by $\sum_{i=1}^N   (X_t^{i}-d_i)^+$, or the shortfall
$-\sum_{i=1}^N   (X_t^{i}-d_i)^-$.

The class ${\cal C}$ of random vectors $\bf Y$ is for instance chosen as ${\bf Y}={\bf w}Y$ where $w^i$'s are weights, and $Y=Y_t$ is defined by
\[
Y_t=y_0+s\left(\rho_0
W^0_{t}+\sqrt{1-\rho_0^{2}}W_{t}\right),
\]
and $W_t$ is a Brownian motion independent of $W^0_t$. This reflects that the money which can be allocated to banks can be correlated to the common factor driving the system. In fact the choice of $W_t$ is general ranging from a linear combination of the $W^i$'s to being independent of the $W^i$'s. This situation corresponds to using securities (in this case bonds) as allocation at the future time $t$. To keep the example simple we consider the case with $s=0$, that is $Y_t=y_0$, and constant injection as we did in Section \ref{sec:Gaussian constant injection}.

\subsubsection{ Homogeneous Network}\label{sec:homogenous}

Here we consider the fully homogeneous case where $x_0^i=x_0$, $p_{i,j}=p/N$, $\sigma^i=\sigma$, $\rho_i={\rho}$, $d_i=d$, so that the model becomes
\begin{align*}
&dX_{t}^{i}=\left[\frac{p}{N}\sum_{j=1}^{N}(X_{t}^{j}-X_{t}^{i})\right]dt+\sigma\left({\rho}
dW^0_{t}+\sqrt{1-{\rho}^{2}}dW_{t}^{i}\right),\quad i=1,\cdots,N,
\end{align*}
or
\begin{align}
&dX_{t}^{i}=p\left[\bar{X}_t-X_{t}^{i}\right]dt+\sigma\left({\rho}
dW^0_{t}+\sqrt{1-{\rho}^{2}}dW_{t}^{i}\right),\quad i=1,\cdots,N,\label{correlated coupled homogeneous}
\end{align}
where
\[
\bar{X}_t=\frac{1}{N}\sum_{j=1}^{N}X_{t}^{j}.
\]
In order to apply the results from Section \ref{sec:Gaussian constant injection}, we will need to compute the distribution of a single $X^i_t$. The joint distribution of the $X^i_t$'s being obviously Gaussian, $X^i_t$ is Gaussian with mean $\mu_i$ and variance $\sigma_i^2$ in the notation of Section \ref{sec:Gaussian constant injection}. A straightforward computation from (\ref{correlated coupled homogeneous}) gives:
\begin{align*}
&\mu_i=\mathbb E(X^i_t)=x_0,
\end{align*}
and
\begin{align}\label{variance}
&\sigma_i^2=\sigma^2(1-\rho^2)(1-\frac{1}{N})\left(\frac{1-e^{-2pt}}{2p}\right)+\sigma^2\left(\rho^2+\frac{1-\rho^2}{N}\right)t.
\end{align}
Note that, even though we only consider marginal distributions in the systemic risk measure proposed in Section \ref{sec:Gaussian constant injection}, these marginal distributions depend on the coupled dynamics of the $X^i$, in particular on the parameters $p$ and $\rho$. For instance, one sees that increasing $p$, that is increasing liquidity, would decrease $\sigma_i^2$ (from $\sigma^2t$ for $p=0$ to $\sigma^2\left(\rho^2+\frac{1-\rho^2}{N}\right)t$ for $p=\infty$), and therefore, would decrease systemic risk according with our findings in Section \ref{sec:Gaussian constant injection}.


\subsubsection{ Central Clearing Network}

Here we consider a centralized model where bank $1$ (for instance) plays a clearing role and is related to each of the other banks which are not directly related to each other.
That is $p_{i,j}=p$ if $i=1$ or $j=1$, and $p_{i,j}=0$ if $i\ne 1$ and $j\ne 1$; $x_0^i=x_0$ if $i\ne 1$; $\sigma^i=\sigma$ if  $i\ne 1$; $\rho_i={\rho}$ if $i\ne 1$; $d_i=d$ if  $i\ne 1$. The model becomes
\begin{align*}
&dX_{t}^{1}=p\sum_{j=2}^{N}(X_{t}^{j}-X_{t}^{1})dt+\sigma_c\left({\rho}_c
dW^0_{t}+\sqrt{1-{\rho}_c^{2}}dW_{t}^{1}\right),\\
&dX_{t}^{i}=p(X_{t}^{1}-X_{t}^{i})dt+\sigma\left({\rho}
dW^0_{t}+\sqrt{1-{\rho}^{2}}dW_{t}^{i}\right),\quad i=2,\cdots,N,
\end{align*}
with initial conditions $X^1_0=x^1_0$ and $X^i_0=x_0$ for $i=2,\cdots,N$.
The joint distribution of the $X^i_t$'s is again Gaussian.
Choosing $x_0^1=x_0=c_0/N$, we get
\begin{align*}
&\mu_i=\mathbb E(X^i_t)=\mathbb E(X^1_t)=c_0/N.
\end{align*}
We turn now to the computation of the variances. Applying It\^o's formula we get:
\begin{align*}
&d\mathbb E[(X^i_t)^2]=2p\left(\mathbb E(X^i_tX^1_t)-\mathbb E[(X^i_t)^2]\right)dt+\sigma^2dt,\quad i=2,\cdots,N,\\
&d\mathbb E[(X^1_t)^2]=2p\sum_{j=2}^N\left(\mathbb E(X^1_tX^j_t)-\mathbb E[(X^1_t)^2]\right)dt+\sigma_1^2dt,\\
&d\mathbb E(X^i_tX^1_t)=p\sum_{j=2}^N\left(\mathbb E(X^i_tX^j_t)-\mathbb E(X^i_tX^1_t)\right)dt\\
&+p\left(\mathbb E[(X^1_t)^2]-\mathbb E(X^i_tX^1_t)\right)dt+\sigma\sigma_c\rho\rho_cdt, \quad i=2,\cdots,N,\\
&d\mathbb E(X^i_tX^j_t)=p\left(\mathbb E(X^1_tX^j_t)-\mathbb E(X^i_tX^j_t)\right)dt+p\left(\mathbb E(X^1_tX^i_t)-\mathbb E(X^i_tX^j_t)\right)dt\\
&+\sigma^2\rho^2dt, \quad i\geq 2, j\geq 2, i\ne j.
\end{align*}
By symmetry among the $X^i$'s for $i\geq 2$, we deduce that $\mathbb E[(X^i_t)^2], \mathbb E(X^i_tX^1_t)$, and $\mathbb E(X^i_tX^j_t)$ do not depend on $i\geq 2, j\geq 2$. Accordingly, we define:
\begin{align*}
&\mathbb E[(X^i_t)^2]=v(t), \quad i\geq 2,\\
&\mathbb E(X^i_tX^1_t)=w(t),  \quad i\geq 2,\\
&\mathbb E(X^i_tX^j_t)=\chi (t), \quad i\geq 2, j\geq 2, i\ne j,\\
&\mathbb E[(X^1_t)^2]=v_1(t).
\end{align*}
These functions satisfies the differential system $Y'=pAY+B$
with:
\[
Y=\left(
\begin{array}{c}
v\\
v_1\\
w\\
\chi
\end{array}
\right),\,
A=\left[
\begin{array}{cccc}
-2&0&2&0\\
0&-2(N-1)&2(N-1)&0\\
1&1&-N&N-2\\
0&0&2&-2
\end{array}
\right],
\,
B=\left(
\begin{array}{c}
\sigma^2\\
\sigma_c^2\\
\sigma\sigma_c\rho\rho_c\\
\sigma^2\rho^2
\end{array}
\right).
\]
We are interested in $v(t)$ and $v_1(t)$. Note that by subtracting $x_0^2$ to these four functions it is enough to solve the system with zero initial conditions.
A straightforward but tedious computation shows that:
\begin{align*}
&\sigma_1^2=v_1(t)=\mbox{var}(X^1_t)=  \sigma^2\rho^2 t\\
&+\frac{1}{N}\left[ (\sigma^2+2\sigma\sigma_c\rho\rho_c-3\sigma^2\rho^2)t
+\frac{2}{p}(\sigma_c^2-\sigma\sigma_c\rho\rho_c)
\right]
+{\cal O}(\frac{1}{N^2}),
\end{align*}
and
\begin{align*}
&\sigma_i^2=v(t)=\mbox{var}(X^i_t)= \sigma^2(1-\rho^2)\left(\frac{1-e^{-2pt}}{2p}\right)+ \sigma^2\rho^2 t\\
&+\frac{1}{N}\left[ (\sigma^2+2\sigma\sigma_c\rho\rho_c-3\sigma^2\rho^2)t-
\sigma^2(1-\rho^2)\left(\frac{1-e^{-2pt}}{2p}\right)\right]
+{\cal O}(\frac{1}{N^2}),
\end{align*}
to be compared with the exact formula (\ref{variance}) in order to compare the systemic risk for a fully connected homogenous network with a central clearing network. At order one in $1/N$, the variance is the same but they may differ at order $1/N$.
Writing $\sigma^2+2\sigma\sigma_c\rho\rho_c-3\sigma^2\rho^2=\sigma^2(1-\rho^2) +2\sigma\rho(\sigma_c\rho_c-\sigma\rho)$, we see that the sign of $\sigma_c\rho_c-\sigma\rho$ determines which network is most stable, that is the one with smaller variance according to the conclusion in Section \ref{sec:Gaussian constant injection}.

\subsubsection{Heterogeneous  Networks}

In practical situations, the network will be heterogeneous described by a system like our simplified Gaussian model (\ref{correlated coupled}).
The joint distribution will be fully characterized by the means $\mu_i$'s and by the covariance matrix $Q=[cov (X^i_t,X^j_t)]$ which will depend on the parameters of the model, in particular the preferences $p_{i,j}$ and the individual $\sigma^i$.
In that case, for given coefficients and initial conditions, one will be able to numerically compute the marginal means and variances needed in our systemic risk measures. Doing so, one will obtain the optimal allocation $m=(m_i)_{i=1,\cdots,N}$ and a ranking of the banks with respect to their systemic risk contributions.

\subsubsection{An Example with Random Injections}

In order to illustrate the results from Section \ref{sec:Gaussian random injection}, we consider a system of three banks. If it is fully connected and homogenous as in Section \ref{sec:homogenous}, then in the context of random injections in  Section \ref{sec:Gaussian random injection}, by symmetry, the optimal $m_i$'s will be equal, the $\alpha_i$'s will also be equal and therefore $\alpha_i=0$ because of the constraint  $\sum \alpha_i=0$. Consequently, the injection will simply be constant as studied in Section \ref{sec:homogenous}.

Now we consider the case a heterogenous network with symmetric preferences such that $p_{2,3}=p/2$ and $p_{1,2}=p_{1,3}=0$, equal starting points $x_0^i=x_0$, equal volatility $\sigma^i=\sigma$, and correlation to common noise $\rho_1,\,\rho_2=\rho_3=\rho$. The vector $X=(X^1,X^2,X^3)$ satisfies
\begin{align*}
&dX^1_t= \sigma \left(\rho_1dW^0_{t}
+\sqrt{1-\rho_1^2}dW^1_{t}\right),\\
&dX^2_t= \frac{p}{2}\left(X^3_t-X^2_t\right)dt+\sigma \left(\rho dW^0_{t}
+ \sqrt{1-\rho^2}dW^2_{t}\right),\\
&dX^3_t= \frac{p}{2}\left(X^2_t-X^3_t\right)dt+\sigma \left(\rho dW^0_{t}
+ \sqrt{1-\rho^2}\right)dW^2_{t}.
\end{align*}
The bank $X^1$ is uncoupled with the symmetric network $(X^2,X^3)$.
A straightforward computation shows that the Gaussian pair $(X^1_t, {X^2_t+X^3_t})$ (where we have aggregated $X^2$ and $X^3$) admits the covariance matrix
\begin{align*}
Q=
\sigma^2t\left(
\begin{array}{cc}
1&2\rho\rho_1\\
2\rho\rho_1&2(1+\rho^2)
\end{array}
\right).
\end{align*}
In the following numerical illustration, we take $x_0^i=0, i=1,2,3$, $\sigma^2t=1$, $\rho=.8$ and we vary $\rho_1\in\{-.5, -.2, 0, .2, .5\}$. In  the notation of Section \ref{sec:Gaussian random injection}, this translates into $\sigma_1^2=\sigma^2t=1$, $\sigma_2^2=2\sigma^2t(1+\rho^2)=3.28$ and varying the covariance $2\rho\rho_1\in\{-.8,-.32,0, .32, .8\}$. The results are displayed in Table \ref{table3}
where  we show the values of $m_1, m_2, \alpha$, and the value of the systemic risk measure $m_1+m_2$. In the deterministic column, these values  do not change since they depend only on the marginal distributions and $\alpha=0$ since in that case allocations are deterministic. Then, they can be compared with the values in the case with random allocations where we see that the gain is more pronounced for negative correlation.

\begin{table}
\centering
\begin{tabular}{|c|c|c|c|}
\hline
$2\rho\rho_1\downarrow$&&Deterministic&Random\\
\hline
&$m_1$&0.3486&0.2671\\
&$m_2$&0.6313&0.6347\\
-0.8&$\alpha$&0&2.1413\\
&$m_1+m_2$&0.9799&0.9018\\
\hline
&$m_1$&0.3486&0.2799\\
&$m_2$&0.6313&0.6577\\
-0.32&$\alpha$&0&1.1161\\
&$m_1+m_2$&0.9799&0.9376\\
\hline
&$m_1$&0.3486&0.3062\\
&$m_2$&0.6313&0.6530\\
0&$\alpha$&0&0.8416\\
&$m_1+m_2$&0.9799&0.9592\\
\hline
&$m_1$&0.3486&0.3271\\
&$m_2$&0.6313&0.6414\\
0.32&$\alpha$&0&0.6813\\
&$m_1+m_2$&0.9799&0.9685\\
\hline
&$m_1$&0.3486&0.3436\\
&$m_2$&0.6313&0.6294\\
0.8&$\alpha$&0&.6597\\
&$m_1+m_2$&0.9799&0.9750\\
\hline
\end{tabular}
\caption{Sensitivity with respect to correlation to common noise.}
\label{table3}
\end{table}


\section{Example: Systems on a Finite Probability Space}\label{finite_space} 

We now consider a financial system $\mathbf{X}=(X^1,\cdots,X^N)$ that is defined on a finite probability space
$(\Omega,\mathcal{F},\mathbb{P})$ with $\Omega=\left\{\omega_1,\cdots,\omega_M\right\}$, $\mathcal{F}=2^{\Omega}$, $\mathbb{P}(\omega_i)=p_i\in (0,1)$, $i=1,\cdots,M$. The systemic risk measure we are interested in here is given by
\begin{equation}\label{eq:RMFin}
\rho(\mathbf{X}):=\inf\left\{{\sum_{i=1}^{N} {Y^i}}\ |\  \mathbf{Y}=(Y^1, \cdots,Y^N)\in\mathcal{C}^h\,,\Lambda(\mathbf{X}+\mathbf{Y})\in{\mathbb{A}_{\gamma}}\right\},
\end{equation}
where as in Section~\ref{sec:Gaussian} the acceptance set is $\mathbb{A}_{\gamma}=\left\{Z\in \mathcal{L}^0(\mathbb{R})| \ \mathbb E[Z]\geq -\gamma \right\}$ for $\gamma >0$ and the admissible allocations $\mathcal{C}^h$ are introduced below. The aggregation is defined by
\begin{equation}\label{eq:agfct}
\Lambda(x_1, \cdots,x_N):=\sum_{i=1}^{N} -\exp{(-\alpha_i x_i)} 
\end{equation}
for $\alpha_i >0$, $i=1,\cdots,N$. Compared to the aggregation in Section~\ref{sec:Gaussian}, the aggregation in \eqref{eq:agfct} is more risk averse with respect to bigger losses but also takes benefits of gains into account. 

Due to the finite probability space the computation of the optimal allocation associated to the risk measure \eqref{eq:RMFin} reduces to solving a finite-dimensional system of equations even for most general scenario-dependent allocation. More precisely, let $h:=(h_1,..,h_{k})$ with $0 < h_1 < h_2 < \cdots < h_{k-1} < h_{k} =N$ represent some partition of $\{1,..,N\}$ for a given $k\in\left\{1,\cdots,N\right\}$. We then introduce the following family of allocations:
{\small
\begin{align}\label{sgca}
\mathcal{C}^h&=\biggl\{ \mathbf{Y} \in{\mathcal{L}^{0}(\mathbb{R}^N)} \ |\  \exists\ d=(d_1,\cdots,d_k) \in \mathbb{R}^k \ \textrm{such that} \ \sum_{i=1}^{h_1}Y^i(w_j)=d_1, \nonumber \\
&\sum_{i=h_1+1}^{h_2}Y^i(w_j)=d_{2}, \cdots, 
\sum_{i=h_{k-1}+1}^{N}Y^i(w_j)=d_{k},\, \textrm{for \ } j=1,
\cdots, M \biggr\} \subseteq \mathcal{C}_{\mathbb{R}}.
\end{align}
}
This corresponds to the situation when the regulator is constrained in the way that she cannot distribute cash freely among all financial institutions but only within $k$ subgroups that are induced by the partition $h$. In other words, the risk measure is the sum of $k$ minimal cash amounts $d_1,...,d_k$ determined today, that at time $T$ can be freely allocated within the $k$ subgroups in order to make the system safe. Note that this family spans from deterministic allocations $\mathcal{C}=\mathbb{R}^N$ for $k=N$ to $\mathcal{C}_{\mathbb{R}}$ for $k=1$.

For a given partition $h$ of subgroups one can now explicitly compute a unique optimal allocation $\mathbf{Y}^*$ and the corresponding systemic risk $\rho(\mathbf{X})=\sum_{i=1}^{N}Y^{i,*}$ in \eqref{eq:RMFin} by solving the corresponding Lagrangian system. For better readability of the text we here state the explicit expressions for the following subfamily of allocations
{\small
\begin{align}\label{setC}
\mathcal{C}^r&=\biggl\{ \mathbf{Y} \in{\mathcal{L}^{0}(\mathbb{R})^N}\ |\  \exists\ (c_N, c_{N-1}, \cdots,c_{N-r}) \in \mathbb{R}^{r+1} \ \textrm{such that} \ \sum_{i=1}^{N}Y^i(w_j)=c_N, \nonumber\\
& \sum_{i=1}^{N-1}Y^i(w_j)=c_{N-1}, \cdots,  
 \sum_{i=1}^{N-r}Y^i(w_j)=c_{N-r},\,\textrm{for \ } j=1,
\cdots, M \biggr\} 
\end{align}
}
\noindent for $0\leq r\leq N-1$. 
This class corresponds to $r$ subgroups of size one and one remaining bigger subgroup of size $N-r$, and the two extreme cases are recovered for $r=0$ and $r=N-1$. Note that the summation over the $Y^i$'s of the subgroups in \eqref{setC} has been re-parametrized compared to \eqref{sgca} for the sake of more accessible expressions below.

The following are the optimal solutions when computing the systemic risk measure \eqref{eq:RMFin} with the set $\mathcal{C}=\mathcal{C}^r$ of type \eqref{setC}. The proof is deferred to the appendix. For notational simplicity we denote by $y^k_j:=Y^{k,*}(\omega_j)$ for $k=1,\cdots,N$, $j=1,\cdots,M$ the optimal allocation. The optimal $c_{N-r}, \cdots,c_N$ are given by

\begin{equation}  \label{c_N}
c_{N-r}=-\beta_{N-r}\log(\frac{\gamma}{\alpha_1 \beta_N d_{N-r}}),
\end{equation}
where 
\begin{align*} 
\beta_{N-r}&=\sum_{i=1}^{N-r}
\frac{1}{\alpha_i} \quad \mbox{and} \quad \beta_N=\sum_{i=1}^{N}
\frac{1}{\alpha_i}, \\
d_{N-r}&=\sum_{j=1}^{M}p_j \exp{\left[-\frac{1}{\beta_{N-r}}
\sum_{i=1}^{N-r}X^i(w_j)-\frac{1}{\beta_{N-r}}\sum_{i=1}^{N-r}\frac{1}{
\alpha_i}\log(\frac{\alpha_1}{\alpha_i}) \right]},
\end{align*} 
and by
\begin{align}  
c_{k}& = c_{k-1} - \frac{1}{\alpha_k}\log(\frac{\gamma}{\alpha_k \beta_N K_k}) \nonumber \\
&= c_{N-r} - \sum_{j=N-r+1}^{k} \frac{1}{\alpha_j}\log(\frac{\gamma}{\alpha_j \beta_N K_j}) \nonumber \\
&= -\beta_{N-r}\log(\frac{\gamma}{\alpha_1 \beta_N d_{N-r}}) - \sum_{j=N-r+1}^{k} \frac{1}{\alpha_j}\log(\frac{\gamma}{\alpha_j \beta_N K_j})\label{c_k}
\end{align}
for $k=N-r+1,\cdots,N$, $r \geq 1$, with
\begin{equation*}
K_{k}=\sum_{j=1}^{M}p_j \exp{\left(-\alpha_{k} X^k(w_j) \right)}.
\end{equation*}
In particular the optimal $c_N$ provides the value of the systemic risk measure, i.e.
\begin{equation}\label{rCr}
\rho(\mathbf{X})= -\beta_{N-r}\log(\frac{\gamma}{\alpha_1 \beta_N d_{N-r}}) - \sum_{j=N-r+1}^{N} \frac{1}{\alpha_j}\log(\frac{\gamma}{\alpha_j \beta_N K_j}).
\end{equation}
The optimal allocations are given  by 
\begin{equation}  \label{1.24}
y^1_j=\frac{1}{\alpha_1\beta_{N-r}}\sum_{i=1}^{N-r}X^i(w_j)-X^1(w_j)+\frac{1}{\alpha_1\beta_{N-r}}\sum_{i=1}^{N-r}\frac{1}{\alpha_i}\log(\frac{\alpha_1}{
\alpha_i})+\frac{1}{\alpha_1\beta_{N-r}}c_{N-r} 
\end{equation}
for $j=1,\cdots,M$, 
by

\begin{align} 
y^k_j& =\frac{1}{\alpha_k}\left[\alpha_1 X^1(w_j)-\alpha_k X^k(w_j)-\log(\frac{
\alpha_1}{\alpha_k})+\alpha_1 y^1_j \right] \label{1.23} \\
&=\frac{1}{\alpha_k\beta_{N-r}}\sum_{i=1}^{N-r}X^i(w_j)-X^k(w_j)-\frac{1}{\alpha_k}\log(\frac{
\alpha_1}{\alpha_k}) \nonumber \\
& \quad \ + \frac{1}{\alpha_k\beta_{N-r}}\sum_{i=1}^{N-r}\frac{1}{\alpha_i}\log(\frac{\alpha_1}{
\alpha_i})+ \frac{1}{\alpha_k \beta_{N-r}}c_{N-r} \label{1.23:I}
\end{align}
for all $k=2,\cdots, N-r-1$ and $j\in {1,\cdots, M}$, and by 

\begin{equation}
y^k_j =c_k-c_{k-1}= - \frac{1}{\alpha_k}\log(\frac{\gamma}{\alpha_k \beta_N K_k})   \label{eq:opt}
\end{equation}
for all $k=N-r,\cdots, N$ and $j\in {1,\cdots, M}$.

\begin{remark}
One could extend the above setting further by adding the possibility to limit cross-subsidization in the allocations. This can be done by introducing another constraint into the family \eqref{sgca} of cash allocations:
\begin{align*}
\mathcal{C}^{h,b}&=\biggl\{ \mathbf{Y} \in{\mathcal{L}^{0}(\mathbb{R})^N} | \ \ \ Y^i\geq b_i, i=1,\cdots,N  \,;\\
 & \  \qquad\sum_{i=1}^{h_1}Y^i(w_j)=d_1,
\sum_{i=h_1+1}^{h_2}Y^i(w_j)=d_{2}, \cdots, \sum_{i=h_{k-1}+1}^{N}Y^i(w_j)=d_{k}, \\
& \  \qquad \textrm{for \ } j=1,
\cdots, M,  \textrm{ and\ } d_{t}\in\mathbb{R} \textrm{\ for\ }  t=0,\cdots,k \biggr\} \subseteq \mathcal{C}_{\mathbb{R}} ,
\end{align*}
where $(b_1,...,b_N)\in\mathbb{R}^N$. For example, putting $b:=(0,...,0)$ excludes cash withdrawals from institutions and in this sense doesn't allow for any cross-subsidization. The systemic risk measure and corresponding optimal allocations solution can now be computed by resorting to the Karush Kuhn Tucker conditions (\cite{Boyd}), see the computations in \cite{tesi_pastore}.
 
\end{remark}

\begin{example}
We conclude this section with a numerical example. We consider a system of four banks represented by the random variables $X^1,X^2,X^3$ and $X^4$ on the probability space $(\Omega,\mathcal{F},\mathbb{P})$, where $\Omega=(\omega_1,\omega_2,\omega_3,\omega_4)$, $\mathcal{F}=2^{\Omega}$ and $\mathbb{P}(\omega_1)=0.64$,
$\mathbb{P}(\omega_2)=\mathbb{P}(\omega_3)=0.16$ and $\mathbb{P}(\omega_4)=0.04$. 
We assume that $X^4$ is independent of $X^1, X^2, X^3$, that $X^2$ is comonotone with  $X^1$ and that $X^3$ is countermonotone with  $X^1$. Furthermore 
$X^1(w_1)=X^1(w_3)=100, X^1(w_2)=X^1(w_4)=-50$, $X^2(w_1)=X^2(w_3)=50, X^2(w_2)=X^2(w_4)=-25$, 
$X^3(w_1)=X^3(w_3)=-25, X^3(w_2)=X^3(w_4)=50$ and $X^4(w_1)=X^4(w_2)=50, X^4(w_3)=X^4(w_4)=-25$. We set $\alpha_i=0.3$ for $i=1,\cdots,4$ and $\gamma=50$ and consider the set $\mathcal{C}^r$ defined in \eqref{setC}.

In this setting, Tables \ref{table4} and \ref{table5} reproduce the deterministic allocations and optimal scenario-dependent 
allocations for $r=2$ given by \eqref{1.24},  \eqref{1.23} and \eqref{eq:opt} as well as the systemic risk measure $\rho$ given by \eqref{rCr}. In Table \ref{table6} we provide the systemic risk measures for all $r=0,1,2,3$.
From Table \ref{table6} we note that the maximum and minimum value of $\rho$ are obtained respectively in the
deterministic ($r=3$) and the fully unconstrained scenario-dependent ($r=0$)
cases.
Whenever one groups ($X_1$ and $X_3$) or ($X_2$ and $X_3$),    $\rho$ is substantially reduced
($-0.56$ or $4.44$), compared to the deterministic case ($79.02$), as these
couples of vectors are counter monotone.
Whenever one groups $X_4$ with any of the $X_1$, $X_2$, $X_3$, there is little
difference ($68.36, 63.71, 72.96$) with respect to the deterministic case
($79.02$), as  $X_4$ is independent from the others.
Grouping $X_1$ and $X_2$ has very little effect ($74.48$) compared to the
deterministic case ($79.02$), as $X_1$ and $X_2$ are comonotone.

\begin{table}
\centering
{\begin{tabular}{|c|c|}
\hline
Systemic risk measure& Case\\
\hline
$-26.36$&$r=0$\\
\hline
$-0.56$&$r=2$, $\left\{X_1,X_3\right\}$\\
\hline
$4.44$&$r=2$, $\left\{X_2,X_3\right\}$\\
\hline
$63.71$&$r=2$, $\left\{X_2,X_4\right\}$\\
\hline
$68.36$&$r=2$, $\left\{X_1,X_4\right\}$\\
\hline
$72.96$&$r=2$, $\left\{X_3,X_4\right\}$\\
\hline
$74.48$&$r=2$, $\left\{X_1,X_2\right\}$\\
\hline
$79.02$&$r=3$\\
\hline
\end{tabular}}
\caption{Systemic risk measure.}
\label{table6}
\end{table}

\begin{table}
\centering
{\begin{tabular}{|c|c|}
\hline
Groups: $\{X_1\},\{X_2\}, \{X_3\}, \{X_4)$&Deterministic\\
\hline
$Y_1$&$36.18$\\
$Y_2$&$15.82$\\
$Y_3$&$15.82$\\
$Y_4$&$11.20$\\
\hline
Systemic risk\,$=\sum Y_i$&${\bf 79.02}$\\
\hline
\end{tabular}}
\caption{Case $r=3$.}
\label{table4}
\end{table}

 \begin{table}
\centering
{\tiny \begin{tabular}{|c|c|c|}
\hline
Groups: $\{X_1,X_2\}, \{X_3\},
\{X_4)$&Random&Deterministic\\
\hline
&$Y_1(\omega_1)=Y_1(\omega_3)=11.27$&\\
$Y_1$&$Y_1(\omega_2)=Y_1(\omega_4)=36.23$&\\
&$\mathbb E[Y_1]=6.23$&\\
\hline
&$Y_2(\omega_1)=Y_2(\omega_3)=48.73$&\\
$Y_2$&$Y_2(\omega_2)=Y_2(\omega_4)=11.23$&\\
&$\mathbb E[Y_1]=41.23$&$Y_1+Y_2=47.46$\\
\hline
$Y_3$&&15.82\\
\hline
$Y_4$&&11.20\\
\hline
Systemic risk\,$=\sum Y_i$&&{\bf 74.48}\\
\hline
\hline
Groups: $\{X_1,X_3\}, \{X_2\}, \{X_4)$&&\\
\hline
&$Y_1(\omega_1)=Y_1(\omega_3)=-76.29$&\\
$Y_1$&$Y_1(\omega_2)=Y_1(\omega_4)=36.21$&\\
&$\mathbb E[Y_1]=-53.79$&\\
\hline
&$Y_3(\omega_1)=Y_3(\omega_3)=48.71$&\\
$Y_3$&$Y_3(\omega_2)=Y_3(\omega_4)=-63.79$&\\
&$\mathbb E[Y_3]=26.21$&$Y_1+Y_3=-27.58$\\
\hline
$Y_2$&&15.82\\
\hline
$Y_4$&&11.20\\
\hline
Systemic risk\,$=\sum Y_i$&&{\bf -0.56}\\
\hline
\hline
Groups: $\{X_1,X_4\}, \{X_2\}, \{X_3)$&&\\
\hline
&$Y_1(\omega_1)=-6.64,Y_1(\omega_2)=68.36$&\\
$Y_1$&$Y_1(\omega_3)=-44.14,Y_1(\omega_4)=30.86$&\\
&$\mathbb E[Y_1]=0.86$&\\
\hline
&$Y_4(\omega_1)=43.36,Y_4(\omega_2)=-31.64$&\\
$Y_4$&$Y_4(\omega_3)=80.86,Y_4(\omega_4)=5.86$&\\
&$\mathbb E[Y_4]=35.86$&$Y_1+Y_4=36.72$\\
\hline
$Y_2$&&15.82\\
\hline
$Y_3$&&15.82\\
\hline
Systemic risk\,$=\sum Y_i$&&{\bf 68.36}\\
\hline
\hline
Groups: $\{X_2,X_3\}, \{X_1\}, \{X_4)$&&\\
\hline
&$Y_2(\omega_1)=Y_2(\omega_3)=-58.97$&\\
$Y_2$&$Y_2(\omega_2)=Y_2(\omega_4)=16.03$&\\
&$\mathbb E[Y_2]=-43.97$&\\
\hline
&$Y_3(\omega_1)=Y_3(\omega_3)=16.03$&\\
$Y_3$&$Y(\omega_2)=Y_3(\omega_4)=-58.97$&\\
&$\mathbb E[Y_3]=1.03$&$Y_2+Y_3=-42.94$\\
\hline
$Y_1$&&36.18\\
\hline
$Y_4$&&11.20\\
\hline
Systemic risk\,$=\sum Y_i$&&{\bf 4.44}\\
\hline
\hline
Groups: $\{X_2,X_4\}, \{X_1\}, \{X_3)$&&\\
\hline
&$Y_2(\omega_1)=5.86, Y_2(\omega_2)=43.36$&\\
$Y_2$&$Y_2(\omega_3)=-31.64,Y_2(\omega_4)=5.86$&\\
&$\mathbb E[Y_2]=5.86$&\\
\hline
&$Y_4(\omega_1)=5.85, Y_4(\omega_2)=-31.65$&\\
$Y_4$&$Y(\omega_3)=43.35,Y_4(\omega_4)=5.85$&\\
&$\mathbb E[Y_4]=5.85$&$Y_2+Y_4=11.71$\\
\hline
$Y_1$&&36.18\\
\hline
$Y_3$&&15.82\\
\hline
Systemic risk\,$=\sum Y_i$&&{\bf 63.71}\\
\hline
\hline
Groups: $\{X_3,X_4\}, \{X_1\}, \{X_2)$&&\\
\hline
&$Y_3(\omega_1)=47.98, Y_3(\omega_2)=10.48$&\\
$Y_3$&$Y_3(\omega_3)=10.48,Y_3(\omega_4)=-27.02$&\\
&$\mathbb E[Y_2]=32.98$&\\
\hline
&$Y_4(\omega_1)=-27.02, Y_4(\omega_2)=10.48$&\\
$Y_4$&$Y(\omega_3)=10.48,Y_4(\omega_4)=47.98$&\\
&$\mathbb E[Y_4]=-12.02$&$Y_3+Y_4=20.96$\\
\hline
$Y_1$&&36.18\\
\hline
$Y_2$&&15.82\\
\hline
Systemic risk\,$=\sum Y_i$&&{\bf 72.96}\\
\hline
\end{tabular}}
\caption{Case $r=2$.}
\label{table5}
\end{table}

\end{example}


\appendix
\section{Appendix}

\subsection{Gaussian Case with Random Injections}

We provide here the computations necessary to minimize the function \eqref{Lagrange_gauss}.
We first consider
\begin{align}
&\esp{(X^i +Y^i-d_i)^-} = \esp{(X^i +m_i + \alpha_i I_D -d_i)^-}  \notag \\
& = \esp{(X^i +m_i + \alpha_i -d_i)^- I_D} + \esp{(X^i +m_i -d_i)^-
I_{A^c}}  \notag \\
& = \esp{\left\{(X^i +m_i + \alpha_i -d_i)^- - (X^i +m_i
-d_i)^-\right\}I_D} + \esp{(X^i +m_i -d_i)^- } \label{gauss_1}
\end{align}
for $i=1,\cdots,N$.
To compute \eqref{gauss_1}, we distinguish between the cases $\alpha_i>0$
and $\alpha_i<0$. Note that by the definition of $\mathcal{C}$, we cannot a
priori argue on the sign of $\alpha$. For $\alpha_i>0$, we have that $
\left\{X^i \leq d_i -m_i\right\}= \left\{X^i \leq d_i
-m_i-\alpha_i\right\}\cup \left\{d_i -m_i-\alpha_i<X^i \leq d_i -m_i\right\}$. 
Here we set $A_1:=\left\{X^i \leq d_i -m_i-\alpha_i\right\}$ and $A_2:=\left\{d_i
-m_i-\alpha_i<X^i \leq d_i -m_i\right\}$. Then
\begin{equation*}
(X^i +m_i + \alpha_i -d_i)^- - (X^i +m_i -d_i)^-= -\alpha_i I_{A_1} + (X^i
+m_i -d_i)I_{A_2} ,
\end{equation*}
and
\begin{align*}
&\esp{(X^i +Y^i-d_i)^-} \\
&\quad= -\alpha_i \esp{I_{A_1}I_D}+ \esp{(X^i +m_i -d_i)
I_{A_2}I_D} + \esp{(X^i +m_i -d_i)^- } \\
&\quad = (m_i -d_i)
F_{i,S}(d_i -m_i, d)  -(m_i+\alpha_i -d_i)F_{i,S}(d_i -m_i-\alpha_i, d) \\
& \quad \ +\int_{d_i -m_i-\alpha_i}^{d_i -m_i}\int_{-\infty}^d xf_{i,S}(x,y)
dydx +\esp{(X^i +m_i -d_i)^- },
\end{align*}
where $F_{i,S}$ and $f_{i,S}$ are the joint distribution function and the
density of $(X^i,S)$, respectively. Recall that in our setting $(X^i,S)\sim
N_2(\bar \mu^i, \bar Q^i)$ with mean vector $\bar \mu^i=(\mu^i, \sum_{j=1}^n
\mu_j)$ and covariance matrix 
\begin{equation*}
\bar Q^i=\left( 
\begin{array}{cc}
\sigma^2_i & \sigma^2_i + \sum_{j\neq i}\rho_{i,j} \\ 
\sigma^2_i + \sum_{j\neq i}\rho_{i,j} & \sum_{j=1}^n \sigma^2_j +
\sum_{j,k=1}^n \rho_{j,k} \\ 
& 
\end{array}
\right). 
\end{equation*}
Analogous computations hold in the case $\alpha_i<0$. Summing up, we obtain
that 
\begin{align*}
&\esp{\sum_{i=1}^N(X^i +Y^i-d_i)^-} =\sum_{i=1}^N\esp{(X^i
+Y^i-d_i)^-} \\
&\quad = \sum_{i=1}^N\esp{(X^i +m_i -d_i)^- } \\
& \quad \ + \sum_{i=1}^NI_{\alpha_i \geq 0} \left[(m_i -d_i) F_{i,S}(d_i -m_i, d)  -(m_i+\alpha_i
-d_i)F_{i,S}(d_i -m_i-\alpha_i, d) \right.
\\
& \quad \quad \left.+\int_{d_i -m_i-\alpha_i}^{d_i -m_i}\int_{-\infty}^d
xf_{i,S}(x,y) dydx\right] \\
& \quad + \sum_{i=1}^NI_{\alpha_i <0} \left[(m_i -d_i) F_{i,S}(d_i -m_i, d) -(m_i+\alpha_i -d_i)F_{i,S}(d_i -m_i-\alpha_i, d)
\right. \\
& \quad \quad\left.+\int_{d_i -m_i-\alpha_i}^{d_i -m_i}\int_{-\infty}^d
xf_{i,S}(x,y) dydx\right] 
\end{align*}
\begin{align*}
&= \sum_{i=1}^N\esp{(X^i +m_i -d_i)^- } \\
& \quad \ + \sum_{i=1}^N\left[(m_i -d_i) F_{i,S}(d_i -m_i, d)  -(m_i+\alpha_i
-d_i)F_{i,S}(d_i -m_i-\alpha_i, d) \right.
\\
& \quad \quad \left.+\int_{d_i -m_i-\alpha_i}^{d_i -m_i}\int_{-\infty}^d
xf_{i,S}(x,y) dydx\right] \\
&= \sum_{i=1}^N\esp{(X^i +m_i -d_i)^- } \\
& \quad \ + \sum_{i=1}^{N-1}\left[(m_i -d_i) F_{i,S}(d_i -m_i, d)  -(m_i+\alpha_i
-d_i)F_{i,S}(d_i -m_i-\alpha_i, d) \right.\\
& \quad \quad \left.+\int_{d_i -m_i-\alpha_i}^{d_i -m_i}\int_{-\infty}^d
xf_{i,S}(x,y) dydx\right] \\
& \quad + (m_N -d_N) F_{N,S}(d_N -m_N, d)  -(m_N-\sum_{j=1}^{N-1}\alpha_j
-d_N)F_{N,S}(d_N -m_N \\
& \quad \quad +\sum_{j=1}^{N-1}\alpha_j, d)+\int_{d_N-m_N+\sum_{j=1}^{N-1}\alpha_j}^{d_N -m_N}\int_{-\infty}^d
xf_{N,S}(x,y) dydx,
\end{align*}
where in the last equality we have used the constraint $\sum_{j=1}^{N}\alpha_j=0$.
We now denote by $\mu_i$, $\sigma_i$ the mean and the quadratic variation of $X^i$, $i=1,\cdots,N$, and $\Phi(x)=\int_{+\infty}^x \frac{1}{\sqrt{2\pi}}e^{-t^2/2} dt$. Set  $\bar{f}_{i,S}(x, y)=\int_{-\infty}^y f_{i,S}(x,s) ds$.

 \begin{enumerate}
\item By computing the derivatives with respect to $\alpha_i$, $i=1,\cdots,N-1$, we obtain
$\frac{\partial \phi}{\partial \alpha_i}=0$ if and only if 
{\small
\begin{align*}
0&= 
 \lambda \Big( (m_i+\alpha_i
-d_i)\bar{f}_{i,S}(d_i -m_i-\alpha_i, d) - F_{i,S}(d_i -m_i-\alpha_i, d)  \\
& \ +(d_i -m_i-\alpha_i) \int_{-\infty}^d f_{i,S}(d_i -m_i-\alpha_i,y)
dy \\
& \quad + F_{N,S}(d_N -m_N+\sum_{j=1}^{N-1}\alpha_j, d)\\
& \quad  \ - (m_N-\sum_{j=1}^{N-1}\alpha_j
-d_N)\bar{f}_{N,S}(d_N -m_N+\sum_{j=1}^{N-1}\alpha_j, d)  \\
& \quad  \ \ + (m_N-\sum_{j=1}^{N-1}\alpha_j
-d_N) \int_{-\infty}^d f_{N,S}(d_N -m_N+\sum_{j=1}^{N-1}\alpha_j,y)
dy 
\Big)\\
&=\lambda \left(  F_{N,S}(d_N -m_N+\sum_{j=1}^{N-1}\alpha_j, d)- F_{i,S}(d_i -m_i-\alpha_i, d) \right).
\end{align*}
}
\noindent We then obtain that the equation above has a solution if $\lambda=0$ or when
\begin{equation} \label{eq_F2}
F_{i,S}(d_i -m_i-\alpha_i, d) =F_{N,S}(d_N -m_N+\sum_{j=1}^{N-1}\alpha_j, d)
\end{equation}
for $i=1,\cdots,N-1$.

\item By computing the derivatives with respect to $m_i$, for $i=1,\cdots,N$, we obtain
$\frac{\partial \phi}{\partial m_i}=0$ if and only if 
\begin{align*}
0&=  1+
 \lambda \left( \Phi(\frac{d_i-\mu_i-m_i}{\sigma_i})- (m_i
-d_i)\bar{f}_{i,S}(d_i -m_i, d)   +  F_{i,S}(d_i -m_i, d) \right. \\
& \left. - F_{i,S}(d_i -m_i-\alpha_i, d) +
(m_i+\alpha_i -d_i)\bar{f}_{i,S}(d_i -m_i-\alpha_i, d)  \right. \\
& \left. +(d_i -m_i-\alpha_i) \int_{-\infty}^d f_{i,S}(d_i -m_i-\alpha_i,y) 
 - (d_i -m_i) \int_{-\infty}^d f_{i,S}(d_i -m_i,y) dy \right)\\
&=  1 + \lambda \left( \Phi(\frac{d_i-\mu_i-m_i}{\sigma_i}) +  F_{i,S}(d_i -m_i, d)  -F_{N,S}(d_N -m_N+\sum_{j=1}^{N-1}\alpha_j, d) \right),
\end{align*}
where we have used \eqref{eq_phi},  \eqref{eq_F2} and the notation above. In particular
{\small
\begin{equation}\label{eq_L}
\lambda =-\left( \Phi(\frac{d_N-\mu_N-m_N}{\sigma_N}) +  F_{N,S}(d_N-m_N, d)  -F_{N,S}(d_N -m_N+\sum_{j=1}^{N-1}\alpha_j, d) \right)^{-1},
\end{equation}
}
if the denominator is different from zero.
By \eqref{eq_L} we then obtain 
\begin{align}
&\Phi(\frac{d_i-\mu_i-m_i}{\sigma_i}) +  F_{i,S}(d_i -m_i, d)\nonumber\\
&=\Phi(\frac{d_N-\mu_N-m_N}{\sigma_N}) +  F_{N,S}(d_N-m_N, d),
\end{align}
for $i=1,\cdots,N-1$.
\end{enumerate}

\subsection{Example on a Finite Probability Space}

We give the proof of the optimal solutions \eqref{c_N}-\eqref{eq:opt}  in Section \ref{finite_space}. \\
 \begin{proof}
We restrict the proof to the case $r=0$. The general case can be obtained following the same steps. For
further details on the proof in the general case, we refer to \cite
{tesi_pastore}.  Note that the following computations
apply in any other case when the derivatives of $U_i$ are invertible for
all $i=1,\cdots,N$. We can rewrite the definition of $\rho$ in this
particular setting as follows:  
\begin{equation}  \label{ro_ut}
\rho(\mathbf{X}):=\inf\left\{c \in \mathbb{R}| \ \sum_{i=1}^N Y^i =c, \mathbb E\left[
\sum_{i=1}^N \expon{-\alpha_i (X^i+Y^i)}\right]\leq \gamma\right\}.
\end{equation}
Note that for $r=0$
 we now have $y^N_j:=c - \sum_{i=1}^{N-1} y^i_j$ for $j=1,\cdots,M$. We
compute $\rho$ by using the method of Lagrange multipliers 
to minimize the function 
\begin{align*}
&\phi(c,y_1^1,\cdots,y^1_{M}, \cdots, y^{N-1}_1,\cdots,y^{N-1}_M,\lambda)\\
& = c+ \lambda \left( \sum_{j=1}^M p_j\left[\sum_{k=1}^{N-1} 
\expon{-\alpha_k(X^k(\omega_j)+y^k_j)}\right.\right.\\
&\left. \left.\qquad\qquad + \expon{-\alpha_N(X^N(\omega_j)+c-
\sum_{i=1}^{N-1} y^i_j)}\right] -\gamma\right).
\end{align*}
We have:
\begin{enumerate}
\item \emph{By computing the derivatives with respect to $y^k_j$, $k=1,\cdots,N-1$
, $j=1,\cdots,M$}: $\frac{\partial \phi}{\partial y^k_j}=0$ if and only if
for all fixed $j=1,\cdots,M$ 
\begin{equation}  \label{eq_j1_0}
\alpha_k \expon{-\alpha_k(X^k(\omega_j)+y^k_j)} = \alpha_N\expon{-\alpha_N(X^N(
\omega_j)+c- \sum_{i=1}^{N-1} y^i_j)}.
\end{equation}
This also implies that for all fixed $j=1,\cdots,M$ 
\begin{equation}  \label{eq_j1}
\alpha_k \expon{-\alpha_k(X^k(\omega_j)+y^k_j)} =\alpha_1 
\expon{-\alpha_1(X^1(\omega_j)+y^1_j)}
\end{equation}
for all $k=1,\cdots,N-1$, i.e. 
\begin{equation}  \label{eq_j1_2}
y^k_j = \frac{1}{\alpha_k}\left[\alpha_1X^1(\omega_j) - \alpha_k
X^k(\omega_j)- \log{\frac{\alpha_1}{\alpha_k}}+\alpha_1 y^1_j\right].
\end{equation}
Furthermore, by \eqref{eq_j1_0} we obtain that 
\begin{equation}  \label{eq_j1_3}
y^1_j = \frac{1}{\alpha_1\beta_N}\sum_{i=1}^N X^i(\omega_j) - X^1(\omega_j) + 
\frac{1}{\alpha_1\beta_N}\sum_{i=1}^N\frac{1}{\alpha_i}\log{\frac{\alpha_1}{
\alpha_i}}+ \frac{1}{\alpha_1\beta_N} c,
\end{equation}
where $\beta_N=\sum_{k=1}^{N} \frac{1}{\alpha_k}$ as before.

\item \emph{By computing the derivatives with respect to $c$}: $\frac{\partial \phi
}{\partial c}=0$ if and only if 
\begin{equation}  \label{eq_c1}
0=1 - \alpha_N \lambda \sum_{j=1}^M p_j \expon{-\alpha_N(X^N(\omega_j)+c-
\sum_{i=1}^{N-1} y^i_j)}.
\end{equation}
We can insert \eqref{eq_j1_0} in \eqref{eq_c1} and obtain 
\begin{equation*}
0=1 - \alpha_k\lambda \sum_{j=1}^M p_j \expon{-\alpha_k(X^k(\omega_j)+y^k_j)}
,
\end{equation*}
i.e. 
\begin{equation}  \label{eq_c2}
\sum_{j=1}^M p_j \expon{-\alpha_k(X^k(\omega_j)+y^k_j)} =\frac{1}{
\alpha_k\lambda}
\end{equation}
for all $k=1,\cdots,N-1$.

\item \emph{By computing the derivatives with respect to $\lambda$}: $\frac{
\partial \phi}{\partial \lambda}=0$ if and only if 
\begin{align}
\gamma&=\mathbb E{ \sum_{i=1}^N \expon{-\alpha_i (X^i+Y^i)}}  \notag \\
&=\sum_{j=1}^M p_j\left[\sum_{k=1}^{N-1} \expon{-\alpha_k(X^k(%
\omega_j)+y^k_j)} + \expon{-\alpha_N(X^N(\omega_j)+c- \sum_{i=1}^{N-1}
y^i_j)}\right]. \label{eq_v}
\end{align}
We
now substitute \eqref{eq_c1} and \eqref{eq_c2} in  \eqref{eq_v} and obtain:
\begin{align}
\gamma&=\sum_{j=1}^M p_j\left(\sum_{k=1}^{N-1} \expon{-\alpha_k(X^k(
\omega_j)+y^k_j)}\right.\nonumber\\
& \left.\qquad+ \expon{-\alpha_N(X^N(\omega_j)+c- \sum_{i=1}^{N-1}
y^i_j)}\right)  \notag \\
&=\sum_{k=1}^{N-1} \sum_{j=1}^M p_j \expon{-\alpha_k(X^k(\omega_j)+y^k_j)}\nonumber\\
&\qquad +
\sum_{j=1}^M p_j \expon{-\alpha_N(X^N(\omega_j)+c- \sum_{i=1}^{N-1} y^i_j)} 
\notag \\
&= \sum_{k=1}^{N-1} \frac{1}{\alpha_k\lambda} + \frac{1}{\alpha_N\lambda} 
\quad = \quad \frac{1}{\lambda}\sum_{k=1}^{N} \frac{1}{\alpha_k}.  \label{eq_l}
\end{align}
Hence 
\begin{equation}  \label{eq_l_exp}
\lambda=\frac{1}{\gamma}\sum_{k=1}^{N} \frac{1}{\alpha_k}= \frac{\beta_N}{%
\gamma}.
\end{equation}
\end{enumerate}
We now compute $c$ by inserting \eqref{eq_j1_2}, \eqref{eq_j1_3} and %
\eqref{eq_l_exp} in \eqref{eq_c2} for $k=1$:
\begin{equation}  \label{eq_c3}
e^{-\frac{c}{\beta_N}} \sum_{j=1}^M p_j \expon{-\frac{1}{\beta_N}(%
\sum_{i=1}^N X^i(\omega_j)
+\sum_{i=1}^N\frac{1}{\alpha_i}\log{\frac{\alpha_1}{\alpha_i}})} =\frac{%
\gamma}{\alpha_1\beta_N}.
\end{equation}
Hence the systemic risk measure, i.e. the optimal $c$, is given by 
\begin{equation}  \label{eq_c4}
\rho(\mathbf{X})=c^*=-\beta_N \log{\left[\frac{\gamma}{\alpha_1\beta_N d_N}\right]},
\end{equation}
where $d_N=\sum_{j=1}^M p_j \expon{-\frac{1}{\beta_N}\sum_{i=1}^N
X^i(\omega_j) -
\frac{1}{\beta_N}\sum_{i=1}^N\frac{1}{\alpha_i}\log{\frac{\alpha_1}{%
\alpha_i}}}$. \newline
 By substituting the optimal value
for $c$ in \eqref{eq_j1_2} and \eqref{eq_j1_3} we also obtain  the optimal allocations

\begin{align*}  
y^k_j &= \frac{1}{\alpha_k \beta_N}\sum_{i=1}^N X^i(\omega_j) - 
X^k(\omega_j)-\frac{1}{\alpha_k}\log(\frac{X^j
\alpha_1}{\alpha_k})\\
& \qquad+ 
\frac{1}{\alpha_k\beta_N}\sum_{i=1}^N\frac{1}{\alpha_i}\log{\frac{\alpha_1}{
\alpha_i}} - \frac{1}{\alpha_k}\log{\left[\frac{\gamma}{\alpha_k\beta_N d}\right]}, 
\end{align*}
for $j=1,\cdots,M$ and $k=1,\cdots,N$.
\end{proof}

\

\


\begin{center}
{\bf Acknowledgment}
\end{center}
\noindent We wish to thank Francesco Rotondi for his help with the simulations.\\\\

\bibliographystyle{plainnat}
\bibliography{risksym}

\end{document}